\documentclass{article}

\parskip \medskipamount
\usepackage[T1]{fontenc}
\usepackage{geometry}
\usepackage{titling}
\usepackage{setspace}
\usepackage{xcolor}
\usepackage{authblk}
\usepackage{amsmath,amssymb,amsthm,amsfonts,bm,bbm}

\DeclareMathOperator*{\argmax}{arg\,max}

\usepackage{mathrsfs,stmaryrd}
\usepackage{upgreek}
\usepackage{stmaryrd}
\SetSymbolFont{stmry}{bold}{U}{stmry}{m}{n}
\usepackage{indentfirst}
\usepackage{subdepth}
\usepackage{graphicx}
\usepackage{diagbox}
\usepackage{pdfpages}
\usepackage[latin1]{inputenc}
\usepackage{url}
\usepackage{float}
\usepackage{appendix}
\usepackage{etoolbox}
\usepackage{microtype}
\usepackage[centerlast]{caption}
\usepackage[plainpages=false,hyperfootnotes=false]{hyperref}
\hypersetup{colorlinks=true,urlcolor=blue,linkcolor=red,citecolor=blue}
\usepackage{cleveref}
\crefname{equation}{\hspace{-0.4em}}{\hspace{-0.4em}}
\crefrangeformat{equation}{(#3#1#4)--(#5#2#6)}
\usepackage{bookmark}
\usepackage{framed}
\usepackage{enumitem}
\usepackage{apacite}

\newtheorem{theorem}{Theorem}[section]
\newtheorem{lemma}[theorem]{Lemma}

\newtheorem{remark}[theorem]{Remark}
\newtheorem{definition}[theorem]{Definition}

\newtheorem{example}[theorem]{Example}
\renewenvironment{proof}{\noindent {\bf Proof.}}{\hfill $\Box$}
\allowdisplaybreaks[4]

\hypersetup{
  pdftitle={On stochastic control problems with higher-order moments},
  pdfauthor={Y. K. Wang, J. Z. Liu, A. Bensoussan, K. F. C. Yiu, and J. Q. Wei}
}

\title{On stochastic control problems with higher-order moments}
\date{\vspace{-6ex}}
\author{
    Yike Wang\thanks{School of Finance, Chongqing Technology and Business University, Chongqing 400067, China.},
    Jingzhen Liu\thanks{Corresponding author. China Institute for Actuarial Science, Central University of Finance and Economics, Beijing 100081, China.},
    Alain Bensoussan\thanks{Naveen Jindal School of Management, University of Texas at Dallas, Dallas, TX 75083.},
    Ka-Fai Cedric Yiu\thanks{Department of Applied Mathematics, The Hong Kong Polytechnic University, Hunghom, Kowloon, Hong Kong, China.},
    Jiaqin Wei\thanks{Key Laboratory of Advanced Theory and Application in Statistics and Data Science--MOE, School of Statistics, East China Normal University, Shanghai 200241, China.}.
}

\begin{document}
\maketitle

\begin{abstract}
In this paper, we focus on a class of time-inconsistent stochastic control problems, where the objective function includes the mean and several higher-order central moments of the terminal value of state.
To tackle the time-inconsistency, we seek both the closed-loop and the open-loop Nash equilibrium controls as time-consistent solutions.
We establish a partial differential equation (PDE) system for deriving a closed-loop Nash equilibrium control,
which does not include the equilibrium value function and is different from the extended Hamilton-Jacobi-Bellman (HJB) equations as in Bj\"ork et al. (Finance Stoch. 21: 331-360, 2017).
We show that our PDE system is equivalent to the extended HJB equations that seems difficult to be solved for our higher-order moment problems.
In deriving an open-loop Nash equilibrium control, due to the non-separable higher-order moments in the objective function,
we make some moment estimates in addition to the standard perturbation argument for developing a maximum principle.
Then, the problem is reduced to solving a flow of forward-backward stochastic differential equations.
In particular, we investigate linear controlled dynamics and some objective functions affine in the mean.
The closed-loop and the open-loop Nash equilibrium controls are identical, which are independent of the state value, random path and the preference on the odd-order central moments.
By sending the highest order of central moments to infinity, we obtain the time-consistent solutions to some control problems whose objective functions include some penalty functions for deviation.
\end{abstract}

\noindent {\bf Keywords:} stochastic control, higher-order moment, time-consistent, closed-loop Nash equilibrium control, open-loop Nash equilibrium control
\vspace{3mm}

\noindent {\bf AMS2010 classification:} Primary: 93E20, 91G80; Secondary: 91B08, 49N90
\vspace{3mm}

\section{Introduction}
\label{sec:Introduction}

\noindent Stochastic control theories have been widely studied in the past few decades.
It is a quite important mathematical technic in the field of finance and economics, when a dynamic decision is supposed to be made in a multi-period problem.
Dynamic programming and Bellman principle of optimality, see \cite{Yong-Zhou-1999}, represent the property of optimal control in a large class of dynamic optimization problems:
An optimal control based on the current information is still optimal in future.
However, there are plenty of ``non-standard'' problems, for which the dynamic programming method cannot be directly used.
In those problems, an optimal control for some initial epochs may no longer be optimal for some future epochs in the sense of global maximization/minimization,
which may induce the decision makers to change their approach and to modify their previous optimal plan at every instant.
This phenomenon is known as ``time-inconsistency''.

The studies on time-inconsistent control problems can be traced back to \cite{Strotz-1955,Pollak-1968,Peleg-Yaari-1973,Goldman-1980,Laibson-1997,Laibson-1998}.
In the early years, the research on time-inconsistent control problems can be almost classified into two main categories: 
\begin{itemize}
\item the problems whose objective function includes non-exponential discounting factors, like $J ( t,u ) = \mathbb{E} [ \int_{t}^{T} Q ( t,s ) H ( s, {X}^{u}_{s}, {u}_{s} ) ds + Q ( t,T ) F ( {X}^{u}_{T} ) | \mathcal{F}_{t} ]$;
\item and the problems whose objective function includes non-linear functions of the conditional expectation of state, like $J ( t,u ) = G ( \mathbb{E} [ {X}^{u}_{T} | \mathcal{F}_{t} ] )$.
\end{itemize}
Here $u$ and ${X}^{u}$ respectively denote the control process and its corresponding state process.
Later, the non-exponential discounting problem was extended to the initial-state dependence problem, while the problem with non-linear functions of expectation was extended to the mean-field type control problem.
On the one hand, the initial-state dependence means that some functions for evaluating the future control, such as the integrand in the objective function, explicitly rely on the current time or the current value of state.
See, e.g., \cite[(6.1)]{Bjork-Khapko-Murgoci-2017}.
This initial-state dependence in linear-quadratic (LQ) problems and utility maximization problems has attracted much attention.
The related literature includes, e.g., \cite{Yong-2011,Yong-2012,Ekeland-Mbodji-Pirvu-2012,Zhao-Shen-Wei-2014,Liu-Lin-Yiu-Wei-2020,Zhang-Purcal-Wei-2021}.
On the other hand, in many mean-field type control problems, the conditional expectation $\mathbb{E} [ {X}^{u}_{T} | \mathcal{F}_{t} ]$ also appears in the controlled dynamic equation of ${X}^{u}$. 
More generally, researchers incorporate the distribution of ${X}^{u}_{s}$ into the controlled dynamic equation (known as McKean-Vlasov equation) to make theoretical extensions.
In terms of application, the mean-field type control problems, sometimes combined with LQ problems, have been deeply investigated for several years,
see, e.g., \cite{Hu-Jin-Zhou-2012,Hu-Jin-Zhou-2017,Bensoussan-Frehse-Yam-2013,Yong-2017} and references therein.
As the most celebrated application, the mean-variance dynamic portfolio selection, which can be traced back to \cite{Markowitz-1952},
was pioneered by \cite{Li-Ng-2000,Zhou-Li-2000} and developed by a huge body of literature including \cite{Basak-Chabakauri-2010,Czichowsky-2013,Bjork-Murgoci-Zhou-2014}, etc.

A straightforward extension of mean-variance dynamic portfolio selection is to incorporate higher-order (central) moments into the objective function.
Given that the distribution of the terminal wealth can be determined by all central/original moments if its moment generating function exists,
an objective function including central/origin moments as many as possible might be regarded as an approximate quantification for the preferences on the distribution.

Moreover, objective functions with higher-order central moments can also arise from improving the classical single-period mean-variance model pioneered by \cite{Markowitz-1952}.
As is well known, in the classical model, the investors aim to maximize the expected return when the variance does not exceed a given level,
or to minimize the variance when the expected return does not fall below a given level.
By Pareto efficiency or the effective frontier theory (see also \cite[Section 6.8]{Yong-Zhou-1999} or \cite{Zhou-Li-2000}),
the primal problem can be embedded into a class of unconstrained optimization problems with ``mean-variance utility'',
where the utility is increasing with the mean and decreasing with the variance.
In the same manner, for the maximizer of mean-variance-skewness utility, which is increasing with respect to the mean and decreasing with respect to the variance and the skewness,
the investor cannot further improve the expected return without increasing the variance and the skewness.
Notably, when the mean and the variance are fixed, a larger skewness usually indicates a smaller mode or median, which is not preferred by some investors.
This coincides with the negative risk premium of skewness found by \cite{ChabiYo-2012} in the multi-period capital asset pricing model with higher-order moments.
However, a positive skewness usually arises from a long right tail, which may be preferred by some investors.
This implies that it is also rational to assume that the mean-variance-skewness utility is increasing with the skewness.
The above discussion on benefit and drawback of positive skewness is still valid for other odd-order central moments, since they can also be treated as measures of the symmetry of a return's distribution.

Furthermore, we have the following informal expansion for the expected utility of the random return $X$ (see also \cite{Samuelson-1970,Jean-1971,Scott-Horvath-1980}),
\begin{equation*}
\mathbb{E} [ \Psi (X) ] = \Psi (0) + \Psi' (0) \mathbb{E} [X] + \sum_{j=2}^{n} \frac{1}{j!} {\Psi }^{(j)} ( \mathbb{E} [X] ) \mathbb{E} \big[ ( X - \mathbb{E} [X] )^{j} \big] + \bar{R},
\end{equation*}
where $\mathbb{E}$ denotes the expectation operator, $\Psi $ is a utility function, and $\bar{R}$ is the residual.
The mean and the higher-order central moments are included in the expansion.
From stochastic dominance theories (see, e.g., \cite{Ekern-1980}), $(-1)^{j+1} {\Psi }^{(j)}$ for every $j$ is supposed to be positive,
which provides appropriate assumptions for the monotonicity of the utility with respect to the central moments.

Besides, in the objective function, the combination of higher-order central moments can act as a penalty function for the deviation $X - \mathbb{E} [X]$,
or can generate a utility function of the random return $X$ net of the endogenous reference level $\mathbb{E} [X]$.
For example, for the objective function
\begin{equation}
J = \mathbb{E} [X] - \sum_{j=2}^{\infty } \frac{ ( - c )^{j} }{j!} \mathbb{E} \big[ ( X - \mathbb{E} [X] )^{j} \big]
  = \mathbb{E} [X] - \mathbb{E} \big[ {e}^{ - c ( X - \mathbb{E} [X] ) } \big] + 1, \quad c > 0,
\label{eq:mean-exponential utility :eq}
\end{equation}
the function $S(x) = \exp ( - c x ) - 1$ provides a relatively small reward (resp. large penalty) for positive (resp. negative) deviation from the expected return.
Also, the abovementioned objective function $J$ can be regarded as a weighted average of the mean $\mathbb{E} [X]$ for a risk-neutral preference on $X$,
and the constant absolute risk aversion (CARA) utility function of $X - \mathbb{E} [X]$, whose CARA coefficient is $c$.
The inspiration about endogenous reference level can be also found in the study on, e.g., consumer's habit formation in behavioral finance (see also \cite{Liu-Lin-Yiu-Wei-2020}).
The related results for the above-mentioned objective function can be found in \Cref{subsubsec: Limiting case I: mean-exponential objective function} in this paper.

In the last few decades, there have been several studies devoted to the portfolio selection problems with higher-order moments.
Researchers proceeded with static (namely, single-period) mean-variance-skewness problem.
\cite{Konno-Suzuki-1995} found that the ``efficient frontier'' for the mean-variance-skewness problem is non-concave (see Figure 1 for its 3D image and Figure 2 for its 2D section, therein),
which may lead to a significant increase in the complexity of further analysis.
At first, the 3D efficient frontier has a flat area and many kinks, where the solution switches from a corner point to an interior point or vice versa.
When drawing a ``capital market surface'' as an analog to capital market line, which should be a ruled surface tangent to the 3D efficient frontier,
one can find that the intersection of the 3D efficient frontier and the capital market surface is sometimes discontinuous.
Furthermore, it is difficult to derive an explicit solution for mean-variance-skewness utility maximization over the capital market surface, since the indifference surface of utility is not necessarily convex.
\cite{DeAthayde-Flores-2004} considered the variance minimization problem with mean-skewness constraints to determine the efficient frontier.
The result is not always the same as that of the skewness maximization problem with mean-variance constraints or the mean maximization problem with variance-skewness constraints.
This also indicates the non-concavity of the efficient frontier.
To handle the non-concavity, many researchers have tried to obtain the optimal solution by specifying the objective function, by considering some multi-object problems, or by designing some numerical algorithms.
See, e.g., \cite{Lai-1991,Konno-Shirakawa-Yamazaki-1993,Joro-Na-2006,Briec-Kerstens-Jokung-2007,Harvey-Liechty-Liechty-Muller-2010,Boissaux-Schiltz-2013,Landsman-Makov-Shushi-2020,Mehlawat-Gupta-Khan-2021,Zhen-Chen-2022}.

For dynamic maximization problems in multi-period or continuous-time models,
the non-concavity of efficient frontier, or the non-convexity of optimization problem, makes it difficult to derive the value function from the associated Hamilton-Bellman-Jacobi (HJB) equation.
The stochastic maximum principle has also been considered (see \cite{Boissaux-Schiltz-2010}), but it is still quite hard to derive an explicit expression of the maximizer/minimizer.
We should note that unlike the systematic skewness and co-skewness of yield rates of all assets emphasized in \cite{Kraus-Litzenberger-1976},
it is the total skewness of terminal wealth that the objective function of a dynamic optimization problem includes.
In fact, the geometric Brownian motion, or the Brownian drift-diffusion model in general, has already provided the multivariate distribution of the long-term yield rates of those assets.

Rather than a maximizer/minimizer for the objective function, we consider time-consistent solutions for our higher-order moment problems.
Notably, the model studied in this paper is fairly general, although it arises from the portfolio problems.
On the one hand, following the definition as in \cite{Bjork-Murgoci-2014,Bjork-Khapko-Murgoci-2017,Wang-Zheng-2021},
we try to seek an optimal feedback control only against some spike variations in form of feedback function,
which is known as the closed-loop Nash equilibrium control, or namely, the Nash equilibrium control law.
The problem is reduced to solving an extended HJB system,
which contains a Bellman equation for the equilibrium value function and some partial differential equations (PDEs) known as the Feynman-Kac formula for original moments.
On the other hand, referring to \cite{Hu-Jin-Zhou-2012,Hu-Jin-Zhou-2017}, we try to seek an open-loop Nash equilibrium control, which is optimal against some spike variations in form of stochastic process.
A standard perturbation argument from \cite[Chapter 3]{Yong-Zhou-1999} is adopted to develop a sufficient maximum principle.
And then the solution can be derived from a flow of forward-backward stochastic differential equations (FBSDEs).

The main contributions of this paper are as follows.
\emph{Firstly}, we establish a PDE system for deriving a closed-loop Nash equilibrium control.
The PDE system does not include the so-called equilibrium value function and is different from the well-known extended HJB equations as in \cite{Bjork-Khapko-Murgoci-2017}.
After using a sequence of multipliers to connect the PDEs arising from the Feynman-Kac formula for original moments,
we establish a verification theorem to show the sufficiency of a feedback control to the optimality.
Our innovative methodology is similar to that in \cite{Wang-Zheng-2021}, where a generalized equilibrium value function is introduced which reduces to the classical equilibrium value function in a diagonal manner.
However, using the system given by \cite{Wang-Zheng-2021} is not appropriate to solve our problem, where the terminal value of the state and its mean and higher-order moments are not separable.
For our approach, when handling the problems whose objective function can be re-expressed as a linear combination of mean and higher-order central moments,
it is not necessary to solve any Bellman equation or PDE to arrive at the time-consistent solution.
\emph{Secondly}, we derive a sufficient maximum principle with a flow of FBSDEs in seeking open-loop Nash equilibrium controls.
For a linear control model, the maximum principle can be further simplified.
Notably, compared with the classical cases, e.g., in \cite{Peng-1990,Yong-Zhou-1999,Buckdahn-Djehiche-Li-2011,Hu-Jin-Zhou-2012},
the estimation of more expansion terms needs to be investigated in the perturbation argument, due to the non-separable higher-order moments in the objective function.
\emph{Thirdly}, we consider a particular case with a linear control model, where the objective function is affine in the mean.
Both categories of time-consistent solutions, closed-loop and open-loop, are derived.
It is interesting to find that they are identical, and they are independent of both the state and random path.
Moreover, the preference on odd-order central moments does not affect the solution.
\emph{Finally}, we consider several limiting cases where the objective functions are linear combinations of mean and central moments and the highest order of the central moments tends to infinity.
The closed-form solution in each case is derived from solving an algebraic equation rather than a differential equation.
Moreover, by a heuristic approach with the use of Fourier cosine transform, we obtain a candidate time-consistent solution for a fairly general even penalty function.

The rest of this paper is organized as follows.
In \Cref{sec:Model formulation}, we introduce the problem with higher-order central moments in a general formulation.
In \Cref{sec:Sufficient condition of CNEC: extended HJB-PDE}, an extended HJB system with a verification theorem for closed-loop Nash equilibrium control is studied.
In \Cref{sec:Sufficient condition of ONEC: maximum principle}, we derive a sufficient maximum principle for open-loop Nash equilibrium control, and we reduce the problem to solving a flow of FBSDEs.
In \Cref{sec:Particular case: linear controlled SDE}, we apply our results to solve several particular linear control problems, and show their state/path-independent solutions.
In \Cref{sec:Concluding remark}, we make some concluding remarks.

\section{Model and problem formulation}
\label{sec:Model formulation}

Let $T$ denote a fixed finite time-horizon, $( \Omega, \mathcal{F}, \mathbb{P} )$ denote a complete probability space, on which a one-dimensional standard Brownian motion $W := \{ {W}_{t} \}_{ t \in [ 0,T ] }$ is defined,
and $\mathbb{E}$ be the expectation operator.
Let $\mathbb{F} := \{ \mathcal{F}_{t} \}_{ t \in [ 0,T ] }$ be the right-continuous, completed natural filtration generated by $W$,
and $\mathbb{E}_{t} [ \cdot ] := \mathbb{E} [ \cdot | \mathcal{F}_{t} ]$ for short.
For the sake of brevity,
in this paper we usually omit the statement $\mathbb{P}$-a.s. (namely, almost surely with respect to $\mathbb{P}$) for equalities and inequalities involving conditional expectations or It\^o integrals,
unless otherwise noted.
We suppose that the dynamics of the real-valued state process $\{ {X}_{t} \}_{ t \in [ 0,T ] }$ evolve as the following fairly general controlled stochastic differential equation (SDE),
\begin{equation}
d {X}_{s} = b ( s, {X}_{s}, {u}_{s} ) ds + \sigma ( s, {X}_{s}, {u}_{s} ) d {W}_{s}, \quad
  {X}_{0} = {x}_{0},
\label{eq:initial controlled SDE :eq}
\end{equation}
and the primal object is to maximize
\begin{equation*}
J ( 0, {x}_{0}, u )
= \mathbb{E} \Big[ \Psi \big( 0, {x}_{0}, {X}_{T}, \mathbb{E} [ {X}_{T} ], \mathbb{E} \big[ ( {X}_{T} - \mathbb{E} [ {X}_{T} ] )^{2} \big], \ldots,
                              \mathbb{E} \big[ ( {X}_{T} - \mathbb{E} [ {X}_{T} ] )^{n} \big] \big) \Big]
\end{equation*}
for a fixed integer $n \ge 2$.
When the time goes to $t$ and the state moves to $x$, or namely, for the initial pair $( t,x )$,
it is supposed to maximize the objective function
\begin{equation}\begin{aligned}
J ( t,x,u ) := \mathbb{E}_{t} \Big[ \Psi \big( & t, x, {X}^{t,x,u}_{T}, \mathbb{E}_{t} [ {X}^{t,x,u}_{T} ], \mathbb{E}_{t} \big[ ( {X}^{t,x,u}_{T} - \mathbb{E}_{t} [ {X}^{t,x,u}_{T} ] )^{2} \big], \\
                                               & \ldots, \mathbb{E}_{t} \big[ ( {X}^{t,x,u}_{T} - \mathbb{E}_{t} [ {X}^{t,x,u}_{T} ] )^{n} \big] \big) \Big],
\end{aligned}
\label{eq:objective function :eq}
\end{equation}
where $\{ {X}^{t,x,u}_{s} \}_{ s \in [ t,T ] }$ is given by the following controlled SDE:
\begin{equation}
d {X}^{t,x,u}_{s} = b ( s, {X}^{t,x,u}_{s}, {u}_{s} ) ds + \sigma ( s, {X}^{t,x,u}_{s}, {u}_{s} ) d {W}_{s}, \quad
  {X}^{t,x,u}_{t} = x.
\label{eq:controlled SDE :eq}
\end{equation}
To reduce the complexity of derivation, the following basic assumptions are adopted:
\begin{enumerate}
\item $\{ {u}_{t} \}_{ t \in [ 0,T ] }$ is valued in $\mathcal{R}$, a subspace of one-dimensional Euclidean space.
\item The initial condition ${X}_{0} = {x}_{0} \in \mathbb{R}$ and
      the functions $b, \sigma : [ 0,T ] \times \mathbb{R} \times \mathcal{R} \to \mathbb{R}$ and $\Psi: [ 0,T ] \times \mathbb{R} \times \mathbb{R} \times \mathbb{R}^{n} \to \mathbb{R}$ are given.
\item $b ( t,x,u )$ and $\sigma ( t,x,u )$ are uniformly continuous, and twice differentiable in $x$.
      Moreover, $( b, \sigma )$ and their partial derivatives $( {b}_{x}, {b}_{xx}, {\sigma }_{x}, {\sigma }_{xx} )$ are Lipschitz continuous in $( x,u )$,
      for which the Lipschitz continuity parameters are uniformly bounded for all $t \in [ 0,T ]$.
\item $\Psi ( t,x,y, \vec{z} )$ is continuously differentiable in $\vec{z} := ( {z}_{1}, \ldots, {z}_{n} )$ and twice differentiable in $y$.
      Moreover, the functions $( \Psi, {\Psi }_{y}, {\Psi }_{yy}, {\Psi }_{ {z}_{1} }, \ldots, {\Psi }_{ {z}_{n} } )$ are Lipschitz continuous in $( y, \vec{z} )$,
      for which the Lipschitz continuity parameters are uniformly bounded for all $( t,x )$;
      and $| \Psi ( \cdot, \cdot, 0, \vec{0} ) |$ is uniformly bounded.
\end{enumerate}

\begin{definition}\label{def:admissible control}
A stochastic process $u = \{ {u}_{t} \}_{ t \in [ 0,T ] }$ is named an ``admissible control'', 
if $u$ is $\mathcal{R}$-valued and $\mathbb{F}$-adapted, $\mathbb{E} [ \int_{0}^{T} ( | b ( t, 0, {u}_{t} ) |^{2} + | \sigma ( t, 0, {u}_{t} ) |^{2} ) dt ] < \infty $,
and the $\mathbb{F}$-adapted strong solution of the corresponding controlled SDE \cref{eq:initial controlled SDE :eq} satisfies $\mathbb{E} [ | {X}_{T} |^{2n - 2 + \epsilon } ] < \infty $ for some $\epsilon > 0$.
Furthermore, a deterministic function $\hat{u}: [ 0,T ] \times \mathbb{R} \to \mathcal{R}$ is named an ``admissible feedback control'', 
if $u$ given by ${u}_{t} = \hat{u} ( t, {X}^{ 0, {x}_{0}, u }_{t} )$ is an admissible control 
and $\hat{u} ( t,x )$ is Lipschitz continuous in $x$, for which the Lipschitz continuity parameter is uniformly bounded for all $t \in [ 0,T ]$. 
\end{definition}

\begin{remark}
Corresponding to any admissible control or admissible feedback control, \cref{eq:initial controlled SDE :eq} admits a unique $\mathbb{F}$-adapted strong solution, for which the theory can be found in \cite[Section 1.6.1]{Yong-Zhou-1999}.
The additional condition $\mathbb{E} [ | {X}_{T} |^{2n - 2 + \epsilon } ] < \infty $ is employed for the case ${\Psi }_{ y {z}_{n} } \ne 0$ in \Cref{sec:Sufficient condition of ONEC: maximum principle}.
If ${\Psi }_{ y {z}_{m} } \equiv 0$ for all $m \ge n / 2$, then $\mathbb{E} [ | {X}_{T} |^{2n - 2 + \epsilon } ] < \infty $ can be weakened to $\mathbb{E} [ | {X}_{T} |^{n} ] < \infty $.
\end{remark}

Let $\mathcal{U}_{0}$ denote the set of all the admissible controls and $\mathcal{U}$ denote the set of all the admissible feedback controls.
With a slight abuse of notation, corresponding to $\hat{u} \in \mathcal{U}$, we let $\{ {X}^{ t, x, \hat{u} }_{s} \}_{ s \in [ t,T ] }$ and $J ( t, x, \hat{u} )$ be the state process and the objective function, respectively.
In addition, for the sake of brevity, we let ${z}_{0} := 1$, ${\partial }_{x} := \frac{\partial }{\partial x}$,
${f}_{ \vec{z} } := ( {f}_{ {z}_{1} }, \ldots, {f}_{ {z}_{n} } )$ for any $f$ continuously differentiable in $\vec{z}$,
${f}_{ \vec{z} \vec{z} }$ denote the Hessian matrix of $f$ with respect to $\vec{z}$ for any $f$ twice differentiable in $\vec{z}$,
and $\langle \cdot, \cdot \rangle$ denote the scalar product (or namely, inner product) of two row vectors;
and employ the infinitesimal operator $\mathcal{D}_{\zeta }$ by
\begin{equation*}
\mathcal{D}_{\zeta } f ( t,x ):= {f}_{t} ( t,x ) + {f}_{x} ( t,x ) b ( t,x, \zeta ) + \frac{1}{2} {f}_{xx} ( t,x ) \big( \sigma ( t,x, \zeta ) \big)^{2}
\end{equation*}
for $\zeta \in \mathcal{R}$ and any bivariate function $f$ continuously differentiable in the first argument and twice differentiable in the second argument.

It is well-known that the dynamic maximization problems,
maximizing \cref{eq:objective function :eq} subject to \cref{eq:controlled SDE :eq} over all admissible controls $u$ (or all $\hat{u} \in \mathcal{U}$) for all initial pairs $( t,x )$, are usually time-inconsistent,
due to the fairly general initial-pair-dependent function $\Psi ( t,x, \cdot, \cdot )$ and the central moments in the objective function.
In other words, any maximizer $u$ for some initial pair $( t,x )$ may not realize the maximization of \cref{eq:objective function :eq} for the initial pair $( s, {X}^{t,x,u}_{s} )$.
For example, \cite[Section IV]{Strotz-1955} illustrates the time-consistency in a non-exponential discounting problem.
In terms of the mean-variance problem in \cite{Zhou-Li-2000}, mirroring the derivation therein one can find that 
\begin{itemize}
\item the maximizer corresponding to the initial pair $( \tau, {X}_{\tau } )$ is represented by $\hat{u} ( \cdot, \cdot ) = f ( \tau, {X}_{\tau }, \cdot, \cdot )$ with a continuously differentiable function $f$,
\item and $f ( \tau, {X}_{\tau }, \cdot, \cdot ) = f ( 0, {X}_{0}, \cdot, \cdot )$ does not necessary hold due to the positive quadratic variation of $\{ f ( t, {X}_{t}, s,y ) \}_{ t \in [ 0, \tau ] }$ with $( s,y )$ fixed.
\end{itemize}
To overcome the time-inconsistency, researchers formulate the problem as a game between the agent and her/his future selves, pioneered by \cite{Strotz-1955}. 
A subgame perfect Nash equilibrium point for this game is regarded as a time-consistent solution, since any deviation from it at any time instant will be worse off and each of the future selves has no incentive to change the equilibrium strategy.
Interested readers can also refer to \cite[Section 2.3]{Bjork-Murgoci-Zhou-2014} for the concept of the game and a ``subgame perfect Nash equilibrium point''.
In this paper, we take this game-theoretic perspective and seek a closed-loop Nash equilibrium control (see also \cite{Bjork-Murgoci-Zhou-2014,Bjork-Khapko-Murgoci-2017}) 
and an open-loop Nash equilibrium control (see also \cite{Hu-Jin-Zhou-2012,Hu-Jin-Zhou-2017}), which are defined as follows.

\begin{definition}\label{def:closed-loop Nash equilibrium control}
$\tilde{u} \in \mathcal{U}$ is a closed-loop Nash equilibrium control (CNEC), if
\begin{equation}
0 \le \liminf_{\varepsilon \downarrow 0}
      \frac{1}{\varepsilon } \big( J ( t, \tilde{X}_{t}, \tilde{u} ) - J ( t, \tilde{X}_{t}, \tilde{u}^{ t, \varepsilon, \zeta } ) \big), \quad \mathbb{P}-a.s., ~ a.e. ~ t \in [ 0,T ),
\label{eq:closed-loop Nash equilibrium control :eq}
\end{equation}
for $\tilde{X} := {X}^{ 0, {x}_{0}, \tilde{u} }$ and any constant $\zeta \in \mathcal{R}$ satisfying $\tilde{u}^{ t, \varepsilon, \zeta } \in \mathcal{U}$,
where $\tilde{u}^{ t, \varepsilon, \zeta }$ is a spike variation of $\tilde{u}$ given by
$\tilde{u}^{ t, \varepsilon, \zeta } ( s, \cdot ) = \tilde{u} ( s, \cdot ) {1}_{\{ s \notin [ t, t + \varepsilon ) \}} + \zeta {1}_{\{ s \in [ t, t + \varepsilon ) \}}$.
Moreover, corresponding to any fixed CNEC $\tilde{u}$, $V ( t,x ) := J ( t,x, \tilde{u} )$ is called an equilibrium value function (EVF).
\end{definition}

\begin{remark}
Our \Cref{def:closed-loop Nash equilibrium control} for CNEC is a bit weaker than that in \cite{Bjork-Khapko-Murgoci-2017},
while \cite{Bjork-Khapko-Murgoci-2017} used an arbitrary function to make the perturbation (namely, $\zeta = \zeta ( s,y )$).
In fact, our modification is due to the use of Dynkin's formula for, e.g., \cref{eq:mean difference in time :eq},
where the feedback control should be continuous in time.
In other words, replacing the constant $\zeta $ by the function $\zeta ( t,x )$ continuous in $t$, as well as replacing $\mathcal{D}_{\zeta }$ by $\mathcal{D}_{ \zeta ( t,x ) }$,
does not affect the validness of results in this work.
Apart from that, our \Cref{def:closed-loop Nash equilibrium control} for CNEC is also much weaker than the strong equilibrium strategy in \cite{Huang-Zhou-2021,He-Jiang-2021}.
Although the weak equilibrium strategy might correspond to a stationary point rather than a maximum as mentioned in \cite[Remark 3.5]{Bjork-Khapko-Murgoci-2017},
the coming sufficient condition guarantees that $J ( t, \bar{X}_{t}, \bar{u}^{ t, \varepsilon, \zeta } ) - J ( t, \bar{X}_{t}, \bar{u} )$,
even if positive, is at most $o ( \varepsilon )$ (see the proof for \Cref{thm:verification theorem}).
Notably, \cite{He-Jiang-2021} has shown that many classical time-inconsistent control problems, including the mean-variance problem in \cite{Basak-Chabakauri-2010}, have no proper strong equilibrium strategy,
if the perturbation $\zeta $ could be a smooth function.
Therefore, we do not intend to find the strong equilibrium strategy for our problem, but leave this work to our future research.
\end{remark}

\begin{definition}\label{def:open-loop Nash equilibrium control}
$\bar{u} \in \mathcal{U}_{0}$ is an open-loop Nash equilibrium control (ONEC), if
\begin{equation}
0 \le \liminf_{\varepsilon \downarrow 0}
      \frac{1}{\varepsilon } \big( J ( t, \bar{X}_{t}, \bar{u} ) - J ( t, \bar{X}_{t}, \bar{u}^{ t, \varepsilon, \zeta } ) \big), \quad \mathbb{P}-a.s., ~ a.e. ~ t \in [ 0,T ),
\label{eq:open-loop Nash equilibrium control :eq}
\end{equation}
for $\bar{X} := {X}^{ 0, {x}_{0}, \bar{u} }$ and any $\mathcal{F}_{t}$-measurable random variable $\zeta $ satisfying $\bar{u}^{ t, \varepsilon, \zeta } \in \mathcal{U}_{0}$,
where $\bar{u}^{ t, \varepsilon, \zeta }$ is a spike variation of $\bar{u}$ given by $\bar{u}^{ t, \varepsilon, \zeta }_{s} = \bar{u}_{s} + \zeta {1}_{\{ s \in [ t, t + \varepsilon ) \}}$.
\end{definition}

\begin{remark}
Distinct from the substitution used in the spike variation for CNEC, the spike variation for ONEC uses the addition of a perturbation.
This modification helps limit the differences, such as $\sigma ( s, \bar{X}_{s}, \bar{u}_{s} + \zeta ) - \sigma ( s, \bar{X}_{s}, \bar{u}_{s} )$ in the proof for \Cref{lem:estimate for square of mean},
by a multiple of $| \zeta |$,
given that the admissibility of $u$ does not ensure the boundedness of every ${u}_{s}$.
In other words, if we had adopted some adequate assumptions for the growth or boundedness of $( b, \sigma )$ and the integrability of $u$,
the spike variation for ONEC could have been generalized as $\bar{u}^{ t, \varepsilon, \zeta }_{s} = \bar{u}_{s} {1}_{\{ s \notin [ t, t + \varepsilon ) \}} + {\zeta }_{s} {1}_{\{ s \in [ t, t + \varepsilon ) \}}$.
See, e.g., \cite{Peng-1990,Buckdahn-Djehiche-Li-2011,Alia-2019}.
\end{remark}

\section{Sufficient condition for CNEC: extended HJB-PDE}
\label{sec:Sufficient condition of CNEC: extended HJB-PDE}

In this section, we introduce an extended HJB system with several integrality and smoothness conditions for deriving a CNEC.
For any $\hat{u} \in \mathcal{U}$, let ${m}^{ \hat{u} }_{j} ( t,x ) := \mathbb{E}_{t} [ ( {X}^{ t,x, \hat{u} }_{T} )^{j} ]$
and $\vec{m}^{ \hat{u} } := ( {m}^{ \hat{u} }_{1}, \ldots, {m}^{ \hat{u} }_{n} )$.
Without loss of generality, we consider the following fairly general objective function with higher-order original moments,
\begin{equation}
J ( t,x, \hat{u} ) = \mathbb{E}_{t} \big[ \Phi \big( t, x, {X}^{ t,x, \hat{u} }_{T}, \vec{m}^{ \hat{u} } ( t,x ) \big) \big],
\label{eq:objective functional alternative :eq}
\end{equation}
where $\Phi $ has the same properties as of $\Psi $, e.g. $\Phi ( t,x,y, \vec{z} )$ is continuously differentiable in $\vec{z}$ and twice differentiable in $y$.
Obviously, in terms of our problem with higher-order central moments,
\begin{equation}
\Phi ( t, x, y, \vec{z} ) = \Psi \bigg( t,x,y, {z}_{1}, {z}_{2} - {z}_{1}^{2}, {z}_{3} - 3 {z}_{1} {z}_{2} + 2 {z}_{1}^{3}, \ldots, \sum_{j=0}^{n} \binom{n}{j} (-1)^{j} {z}_{1}^{j} {z}_{n-j} \bigg).
\label{eq:objective functional rearrangement for our problem :eq}
\end{equation}
In particular, we have $\Phi ( t,x,y,z, {z}^{2}, \ldots, {z}^{n} ) = \Psi ( t,x,y,z, 0, \ldots, 0 )$.
For the $( \mathbb{F}, \mathbb{P} )$-martingales 
\begin{equation*}
\{ {m}^{ \hat{u} }_{j} ( v, {X}^{ t,x, \hat{u} }_{v} ) = \mathbb{E}_{v} [ ( {X}^{ t,x, \hat{u} }_{T} )^{j} ] \}_{ v \in [ t,T ] } \quad and \quad
\{ \mathbb{E}_{v} [ \Phi ( s, y, {X}^{ t,x, \hat{u} }_{T}, \vec{z} ) ] \}_{ v \in [ t,T ] } 
\end{equation*}
for any $( t,x ), ( s,y, \vec{z} )$ and appropriate $\hat{u}$,
Feynman-Kac formula gives their PDE representation (see, e.g., \cite[Theorem 7.4.1, p. 373]{Yong-Zhou-1999}).
As an analogue, for any $\tilde{u} \in \mathcal{U}$, applying It\^o's rule to ${m}^{(j)} ( v, {X}^{ t,x, \tilde{u} }_{v} )$ and ${U}^{ s,y, \vec{z} } ( v, {X}^{ t,x, \tilde{u} }_{v} )$ produces the following lemma.
The proof is straightforward, so we omit it due to the page limit.

\begin{lemma}\label{lem:Feynman-Kac representation}
Fix $\tilde{u} \in \mathcal{U}$ and $( s,y, \vec{z} ) \in [ 0,T ) \times \mathbb{R} \times \mathbb{R}^{n}$.
Assume that ${U}^{ s,y, \vec{z} }$ and $\vec{m} := ( {m}^{(1)}, \ldots, {m}^{(n)} )$ are the classical solutions of the following PDEs on $[ 0,T ] \times \mathbb{R}$:
\begin{align}
& 0 = \mathcal{D}_{ \tilde{u} ( t,x ) } {U}^{ s,y, \vec{z} } ( t,x ), && s.t. \qquad {U}^{ s,y, \vec{z} } ( T,x ) = \Phi ( s,y,x, \vec{z} ); \label{eq:PDE U :eq}
\\
& 0 = \mathcal{D}_{ \tilde{u} ( t,x ) } {m}^{(j)} ( t,x ), && s.t. \qquad {m}^{(j)} ( T,x ) = {x}^{j}, \quad \forall j = 1,2,\ldots, n, \label{eq:PDE m :eq}
\end{align}
respectively; and for any $( t,x ) \in [ 0,T ) \times \mathbb{R}$,
\begin{align}
& \mathbb{E} \bigg[ \int_{t}^{T} \big| {U}^{ s,y, \vec{z} }_{x} ( v, {X}^{ t,x, \tilde{u} }_{v} )
                                       \sigma \big( v, {X}^{ t,x, \tilde{u} }_{v}, \tilde{u} ( v, {X}^{ t,x, \tilde{u} }_{v}) \big) \big|^{2} dv \bigg] < \infty; \label{eq:integrability of U :eq}
\\
& \mathbb{E} \bigg[ \int_{t}^{T} \big| {m}^{(j)}_{x} ( v, {X}^{ t,x, \tilde{u} }_{v} )
                                       \sigma \big( v, {X}^{ t,x, \tilde{u} }_{v}, \tilde{u} ( v, {X}^{ t,x, \tilde{u} }_{v}) \big) \big|^{2} dv \bigg] < \infty,
  \quad \forall j = 1,2,\ldots, n. \label{eq:integrability of m :eq}
\end{align}
Then, ${U}^{ s,y, \vec{z} } ( t,x ) = \mathbb{E}_{t} [ \Phi ( s,y, {X}^{ t,x, \tilde{u} }_{T}, \vec{z} ) ]$ and $\vec{m} ( t,x ) = \vec{m}^{ \tilde{u} } ( t,x )$.
Furthermore, if $\tilde{u}$ is a CNEC, then $V ( t,x ) = {U}^{ t,x, \vec{m} ( t,x ) } ( t,x )$ is an EVF.
\end{lemma}

Similarly, for fixed $\tilde{u} \in \mathcal{U}$ and $( s,y )\in [ 0,T ) \times \mathbb{R}$, if $\vec{\lambda }^{s,y}$ is the classical solution of
\begin{align}
& 0 = \mathcal{D}_{ \tilde{u} ( t,x ) } \vec{\lambda }^{s,y} ( t,x ),
\quad s.t. \quad \vec{\lambda }^{s,y} ( T,x ) = {\Phi }_{ \vec{z} } \big( s,y, x, \vec{m} ( s,y ) \big),
\label{eq:PDE lambda :eq}
\\
& \mathbb{E} \bigg[ \int_{t}^{T} \langle \vec{\lambda }^{s,y}_{x} ( v, {X}^{ t,x, \tilde{u} }_{v} ), \vec{\lambda }^{s,y}_{x} ( v, {X}^{ t,x, \tilde{u} }_{v} ) \rangle
                                 \big| \sigma \big( v, {X}^{ t,x, \tilde{u} }_{v}, \tilde{u} ( v, {X}^{ t,x, \tilde{u} }_{v} ) \big) \big|^{2} dv \bigg] < \infty,
\label{eq:integrability of lambda :eq}
\end{align}
then $\vec{\lambda }^{s,y} ( t,x ) = \mathbb{E}_{t} [ {\Phi }_{ \vec{z} } ( s,y, {X}^{ t,x, \tilde{u} }_{T}, \vec{m} ( s,y ) ) ]$.
For the sake of notational brevity, we omit the statement of the dependence of $\tilde{u}$ for $( \vec{m}, {U}^{ s,y, \vec{z} }, \vec{\lambda }^{s,y} )$ in what follows, unless otherwise mentioned.
As a main result, the following theorem shows that a CNEC can be derived from a system including the above-mentioned PDEs and square-integrability conditions with an optimality condition.

\begin{theorem}\label{thm:verification theorem}
For a fixed $\tilde{u} \in \mathcal{U}$ and any $( s,y, \vec{z} ) \in [ 0,T ) \times \mathbb{R} \times \mathbb{R}^{n}$,
assume that $( \vec{m}, {U}^{ s,y, \vec{z} }, \vec{\lambda }^{s,y} )$ are the classical solution of the PDEs \cref{eq:PDE m :eq,eq:PDE U :eq,eq:PDE lambda :eq} with the optimality condition
\begin{equation}
\tilde{u} ( t,x ) \in \argmax_{ \zeta \in \mathcal{R} } \{ \mathcal{D}_{\zeta } {U}^{ t,x, \vec{m} ( t,x ) } ( t,x ) + \langle \vec{\lambda }^{t,x} ( t,x ), \mathcal{D}_{\zeta } \vec{m} ( t,x ) \rangle \}
\label{eq:optimality condition :eq}
\end{equation}
and the square-integrability conditions \cref{eq:integrability of U :eq,eq:integrability of m :eq,eq:integrability of lambda :eq}.
Then, $\tilde{u}$ is a CNEC.
\end{theorem}

\begin{proof}
Write ${X}^{ t, \tilde{X}_{t}, \zeta } := {X}^{ t, \tilde{X}_{t}, \hat{u} ( \cdot, \cdot ) }$ with $\hat{u} ( \cdot, \cdot ) \equiv \zeta \in \mathcal{R}$ for short.
This slight abuse of notation is used throughout the rest of this paper, unless otherwise stated.
Based on the previous results for $( \vec{m}, {U}^{ s,y, \vec{z} }, \vec{\lambda }^{s,y} )$, let us proceed with
\begin{equation}
    {m}^{ \tilde{u}^{ t, \varepsilon, \zeta } }_{j} ( t,x )
= \mathbb{E}_{t} [ {m}^{ \tilde{u} }_{j} ( t + \varepsilon, {X}^{ t,x, \zeta }_{ t + \varepsilon } ) ]
= {m}^{ \tilde{u} }_{j} ( t,x ) + \mathcal{D}_{\zeta } {m}^{ \tilde{u} }_{j} ( t,x ) \varepsilon + o ( \varepsilon ),
\label{eq:mean difference in time :eq}
\end{equation}
where the first equality follows from the definition of ${m}^{ \hat{u} }_{j} ( \cdot, \cdot )$ and the tower property (or namely, iterated conditioning)
and the last equality is a straightforward application of Dynkin's formula.
Applying the tower property and Dynkin's formula again yields
\begin{align*}
  \mathbb{E}_{t} [ \Phi ( s,y, {X}^{ t + \varepsilon, {X}^{ t,x, \zeta }_{ t + \varepsilon }, \tilde{u} }_{T}, \vec{z} ) ]
= {U}^{ s,y, \vec{z} } ( t,x ) + \mathcal{D}_{\zeta } {U}^{ s,y, \vec{z} } ( t,x ) \varepsilon + o ( \varepsilon ).
\end{align*}
Then, given the smoothness of $\Phi ( t,x, \cdot, \cdot )$, we have
\begin{align*}
  J ( t,x, \tilde{u}^{ t, \varepsilon, \zeta } )
& = \mathbb{E}_{t} \big[ \Phi \big( t, x, {X}^{ t + \varepsilon, {X}^{ t,x, \zeta }_{ t + \varepsilon }, \tilde{u} }_{T},
                                    \mathbb{E}_{t} [ \vec{m}^{ \tilde{u} } ( t + \varepsilon, {X}^{ t,x, \zeta }_{ t + \varepsilon } ) ] \big) \big] \\
& = J ( t,x, \tilde{u} ) + \big( \mathcal{D}_{\zeta } {U}^{ t,x, \vec{m} ( t,x ) } ( t,x ) + \langle \vec{\lambda }^{t,x} ( t,x ), \mathcal{D}_{\zeta } \vec{m} ( t,x ) \rangle \big) \varepsilon + o ( \varepsilon ) \\
& \le J ( t,x, \tilde{u} ) + o ( \varepsilon ),
\end{align*}
where the last inequality follows from \cref{eq:optimality condition :eq} with the linear combination of the PDEs in \cref{eq:PDE U :eq,eq:PDE m :eq}, that is,
$0 = \mathcal{D}_{ \tilde{u} ( t,x ) } {U}^{ t,x, \vec{m} ( t,x ) } ( t,x ) + \langle \vec{\lambda }^{t,x} ( t,x ), \mathcal{D}_{ \tilde{u} ( t,x ) } \vec{m} ( t,x ) \rangle$.
Consequently, $\liminf_{ \varepsilon \downarrow 0 } \frac{1}{\varepsilon } [ J ( t,x, \tilde{u} ) - J ( t,x, \tilde{u}^{ t, \varepsilon , \zeta } ) ] \ge 0$,
and hence, $\tilde{u}$ is a CNEC.
\end{proof}

In addition to the previous smoothness condition for $( \Psi, \Phi )$, we now assume that $\Phi = \Phi ( t,x,y, \vec{z} )$ is differentiable in $t$ and twice differentiable in $( x, \vec{z} )$.
Then, we are able to establish a Bellman equation for the EVF.

\begin{theorem}[Verification theorem]\label{thm:Bellman equation}
Assume that the following properties hold:
\begin{enumerate}
\item For a fixed $\tilde{u} \in \mathcal{U}$ and any $( s,y, \vec{z} ) \in [ 0,T ) \times \mathbb{R} \times \mathbb{R}^{n}$,
      $( \vec{m}, {U}^{ s,y, \vec{z} }, \vec{\lambda }^{s,y} )$ fulfill the PDEs \cref{eq:PDE U :eq,eq:PDE m :eq,eq:PDE lambda :eq}
      with the square-integrability condition \cref{eq:integrability of U :eq,eq:integrability of m :eq,eq:integrability of lambda :eq}.
\item ${U}^{ s,y, \vec{z} } ( t,x )$ is continuously differentiable in $s$ and twice differentiable in $( y, \vec{z} )$.
\item For $\tilde{U} ( t,x ) := {U}^{ t,x, \vec{m} ( t,x ) } ( t,x )$,
      $V$ is a classical solution of
      \begin{equation}
      \begin{aligned}
      0 = \max_{ \zeta \in \mathcal{R} } \big\{ & \mathcal{D}_{\zeta } V ( t,x ) - \mathcal{D}_{\zeta } \tilde{U} ( t,x ) \\
                                              & + \mathcal{D}_{\zeta } {U}^{ t,x, \vec{m} ( t,x ) } ( t,x ) + \langle \vec{\lambda }^{t,x} ( t,x ), \mathcal{D}_{\zeta } \vec{m} ( t,x ) \rangle \big\},
      \end{aligned} \label{eq:Bellman equation :eq}
      \end{equation}
      whereby $\tilde{u} ( t,x )$ realizes the maximum, with the terminal condition
      \begin{equation}
      V ( T,x ) = \Phi ( T,x,x,x, {x}^{2}, \ldots, {x}^{n} ) \equiv \Psi ( T,x,x,x, 0, \ldots, 0 ). \label{eq:Bellman equation TC :eq}
      \end{equation}
\item For any $( t,x ) \in [ 0,T ) \times \mathbb{R}$ and $f = {V}_{x}, \tilde{U}_{x}$,
      \begin{equation*}
      \mathbb{E} \bigg[ \int_{t}^{T} \Big| f ( v, {X}^{ t,x, \tilde{u} }_{v} ) \sigma \big( v, {X}^{ t,x, \tilde{u} }_{v}, \tilde{u} ( v, {X}^{ t,x, \tilde{u} }_{v}) \big) \Big|^{2} dv \bigg] < \infty.
      \end{equation*}
\end{enumerate}
Then, $\tilde{u}$ is a CNEC, and $V ( t,x ) = J ( t,x, \tilde{u} )$; that is, $V$ is an EVF.
\end{theorem}

\begin{proof}
Plugging the maximizer $\tilde{u} ( t,x )$ with the PDEs in \cref{eq:PDE U :eq,eq:PDE m :eq} into the right-hand side of \cref{eq:Bellman equation :eq} yields
$\mathcal{D}_{ \tilde{u} (t,x) } V ( t,x ) = \mathcal{D}_{ \tilde{u} (t,x) } \tilde{U} ( t,x )$.
Given the square-integrability conditions for $( {V}_{x}, \tilde{U}_{x} )$,
and that $V ( T,x ) = \tilde{U} ( T,x )$ follows from \cref{eq:Bellman equation TC :eq} and the terminal conditions in \cref{eq:PDE U :eq,eq:PDE m :eq},
we can show $V \equiv \tilde{U}$ by applying It\^o's rule to $V ( v, {X}^{ t,x, \tilde{u} }_{v} ) - \tilde{U} ( v, {X}^{ t,x, \tilde{u} }_{v} )$.
Consequently, \cref{eq:Bellman equation :eq} gives \cref{eq:optimality condition :eq}, and then we conclude that $\tilde{u}$ is a CNEC according to \Cref{thm:verification theorem}.
Moreover, due to \Cref{lem:Feynman-Kac representation}, $V ( t,x ) = \tilde{U} ( t,x ) = J ( t,x, \tilde{u} )$ is an EVF.
\end{proof}

Collecting \cref{eq:PDE U :eq,eq:PDE m :eq,eq:PDE lambda :eq,eq:Bellman equation :eq,eq:Bellman equation TC :eq}, we obtain an extended HJB system,
which can be regarded as an extension of that in \cite[Definition 4.1]{Bjork-Khapko-Murgoci-2017}.
Thanks to \Cref{thm:verification theorem}, it is not necessary to solve the Bellman equation \cref{eq:Bellman equation :eq} for deriving a CNEC in our approach.
To end this section, we consider the special case with $\Phi ( \cdot, \cdot, y, \vec{z} ) \equiv \tilde{\Phi } ( y, \vec{z} )$; that is, $\Phi ( t,x,y, \vec{z} )$ is independent of the initial pair $( t,x )$.
The results are used to investigate some particular problems in \Cref{sec:Particular case: linear controlled SDE}.
In this case, with $\vec{m} ( t,x ) = \vec{m}^{ \tilde{u} } ( t,x )$ corresponding to the CNEC $\tilde{u}$,
the row vector $\vec{\lambda }^{s,y} ( t,x ) = \mathbb{E}_{t} [ \tilde{\Phi }_{ \vec{z} } ( {X}^{ t,x, \tilde{u} }_{T}, \vec{m} ( s,y ) ) ]$
and the matrix ${\mu }^{s,y} ( t,x ) := \mathbb{E}_{t} [ \tilde{\Phi }_{ \vec{z} \vec{z} } ( {X}^{ t,x, \tilde{u} }_{T}, \vec{m} ( s,y ) ) ]$,
plugging the result of applying the chain rule to $\mathcal{D}_{\zeta } \tilde{U} ( t,x )$ into the right-hand side of \cref{eq:Bellman equation :eq} yields the Bellman equation
\begin{equation}
\begin{aligned}
0 = \max_{ \zeta \in \mathcal{R} }
    \Big\{ & \mathcal{D}_{\zeta } V ( t,x )
           - \langle \vec{\lambda }^{t,x}_{x} ( t,x ), \vec{m}_{x} ( t,x ) \rangle \big( \sigma ( t,x, \zeta ) \big)^{2} \\
         & + \frac{1}{2} \langle \vec{m}_{x} ( t,x ) {\mu }^{t,x} ( t,x ), \vec{m}_{x} ( t,x ) \rangle \big( \sigma ( t,x, \zeta ) \big)^{2} \Big\},
\end{aligned}
\label{eq:Bellman equation reduced :eq}
\end{equation}
whereby $\tilde{u} ( t,x )$ realizes the maximum on the right-hand side.

\section{Sufficient condition for ONEC: maximum principle}
\label{sec:Sufficient condition of ONEC: maximum principle}

In this section, we refer to the spike variation method in \cite[pp. 124--140]{Yong-Zhou-1999} to develop a sufficient maximum principle for ONEC problem.
To verify that $\bar{u} \in \mathcal{U}_{0}$ is an ONEC, it suffices to prove that
$J ( t, \bar{X}_{t}, \bar{u} ) - J ( t, \bar{X}_{t}, \bar{u}^{ t, \varepsilon, \zeta } ) \ge o ( \varepsilon )$
for any $t$ and any spike variation $\bar{u}^{ t, \varepsilon, \zeta } \in \mathcal{U}_{0}$ of $\bar{u}$.
We fix $t \in [ 0,T )$ and suppose that the values of the $\mathcal{F}_{t}$-measurable random variables $( \bar{X}_{t}, \zeta )$ are given conditioned on $\mathcal{F}_{t}$,
and then write ${X}^{\varepsilon } := {X}^{ t, \bar{X}_{t}, \bar{u}^{ t, \varepsilon, \zeta } }$ for short.
In addition, for the sake of brevity,
we write $\bar{f} (s) := f ( s, \bar{X}_{s}, \bar{u}_{s} )$ and $\delta f(s) := f ( s, \bar{X}_{s}, \bar{u}_{s} + \zeta ) - \bar{f} (s)$ for $f = b, {b}_{x}, {b}_{xx}, \sigma, {\sigma }_{x}, {\sigma }_{xx}$.
To derive an expansion of $J ( t, \bar{X}_{t}, \bar{u} ) - J ( t, \bar{X}_{t}, \bar{u}^{ t, \varepsilon, \zeta } )$ with respect to $\varepsilon $,
let us proceed with the first-order and the second-order variational equations
\begin{equation}
\left\{ \begin{aligned}
d {y}^{\varepsilon }_{s} & = \bar{b}_{x} (s) {y}^{\varepsilon }_{s} ds
                           + \big( \bar{\sigma }_{x} (s) {y}^{\varepsilon }_{s} + {1}_{\{ s \in [ t, t + \varepsilon ) \}} \delta \sigma (s) \big) d {W}_{s}, \quad
  {y}^{\varepsilon }_{t} = 0; \\
d {z}^{\varepsilon }_{s} & = \Big( \bar{b}_{x} (s) {z}^{\varepsilon }_{s}
                                 + \frac{1}{2} \bar{b}_{xx} (s) ( {y}^{\varepsilon }_{s} )^{2}
                                 + {1}_{\{ s \in [ t, t + \varepsilon ) \}} \delta b(s) \Big) ds \\
                   & \quad + \Big( \bar{\sigma }_{x} (s) {z}^{\varepsilon }_{s}
                                 + \frac{1}{2} \bar{\sigma }_{xx} (s) ( {y}^{\varepsilon }_{s} )^{2}
                                 + {1}_{\{ s \in [ t, t + \varepsilon ) \}} \delta {\sigma }_{x} (s) {y}^{\varepsilon }_{s} \Big) d {W}_{s}, \quad
  {z}^{\varepsilon }_{t} = 0.
\end{aligned} \right.
\label{eq:variational equations :eq}
\end{equation}
Following the same line of proof as for \cite[Theorem 4.4, p. 128]{Yong-Zhou-1999}, we can show that for any positive integer $k$,
\begin{equation}
\left\{ \begin{aligned}
& \begin{aligned}
  &  \sup_{ s \in [ t,T ] } \mathbb{E}_{t} [ | {X}^{\varepsilon }_{s} - \bar{X}_{s} |^{2k} ] = O ( {\varepsilon }^{k} ),
  && \sup_{ s \in [ t,T ] } \mathbb{E}_{t} [ | {y}^{\varepsilon }_{s} |^{2k} ] = O ( {\varepsilon }^{k} ), \\
  & \sup_{ s \in [ t,T ] } \mathbb{E}_{t} [ | {z}^{\varepsilon }_{s} |^{2k} ] = O ( {\varepsilon }^{2k} ),
  && \sup_{ s \in [ t,T ] } \mathbb{E}_{t} [ | {X}^{\varepsilon }_{s} - \bar{X}_{s} - {y}^{\varepsilon }_{s} |^{2k} ] = O ( {\varepsilon }^{2k} ),
  \end{aligned} \\
& \sup_{ s \in [ t,T ] } \mathbb{E}_{t} [ | {X}^{\varepsilon }_{s} - \bar{X}_{s} - {y}^{\varepsilon }_{s} - {z}^{\varepsilon }_{s} |^{2k} ] = o ( {\varepsilon }^{2k} ),
\end{aligned} \right.
\label{eq:moment estimates :eq}
\end{equation}
and hence, for $i = 0, 1, \ldots, n$,
\begin{equation}
  \mathbb{E}_{t} [ ( {X}^{\varepsilon }_{T} )^{i} ]
= \mathbb{E}_{t} [ ( \bar{X}_{T} )^{i} ]
+ i \mathbb{E}_{t} [ ( \bar{X}_{T} )^{i-1} ( {y}^{\varepsilon }_{T} + {z}^{\varepsilon }_{T} ) ]
+ \binom{i}{2} \mathbb{E}_{t} [ ( \bar{X}_{T} )^{i-2} ( {y}^{\varepsilon }_{T} )^{2} ] + o ( \varepsilon ).
\label{eq:difference in moments of spike variation :eq}
\end{equation}

From \cref{eq:moment estimates :eq}, we have
$\sup_{ s \in [ t,T ] } | \mathbb{E}_{t} [ {y}^{\varepsilon }_{s} ] |^{2} \le \sup_{ s \in [ t,T ] } \mathbb{E}_{t} [ | {y}^{\varepsilon }_{s} |^{2} ] = O ( \varepsilon )$,
which implies that $\mathbb{E}_{t} [ {y}^{\varepsilon }_{s} ]$ and $\sqrt{\varepsilon }$ might be infinitesimals of the same order.
As a consequence, the terms in the form of $\mathbb{E}_{t} [ A {y}^{\varepsilon }_{T} ] \mathbb{E}_{t} [ B {y}^{\varepsilon }_{T} ]$
in the expansion of $J ( t, \bar{X}_{t}, \bar{u}^{ t, \varepsilon, \zeta } ) - J ( t, \bar{X}_{t}, \bar{u} )$ would seriously interfere with the establishment of adjoint equations.
Fortunately, the following lemma ensures that the above-mentioned dilemma does not occur.

\begin{lemma}\label{lem:estimate for square of mean}
Fix $\tau \in ( t,T ]$. Suppose that $\xi $ is an $( \varepsilon, \zeta )$-independent $\mathcal{F}_{\tau }$-measurable random variable, and
$\mathbb{E} [ | \xi |^{ 2 + \varrho } ] < \infty $ for some $\varrho > 0$.
Then, $| \mathbb{E}_{t} [ \xi {y}^{\varepsilon }_{\tau } ] |^{2} = o ( \varepsilon )$, and $\sup_{ \tau \in [ t,T ] } | \mathbb{E}_{t} [ {y}^{\varepsilon }_{\tau } ] |^{2} = o ( \varepsilon )$.
\end{lemma}

\begin{proof}
Inspired by the proof for \cite[Proposition 3.1]{Buckdahn-Djehiche-Li-2011}, we introduce the factor
\begin{equation*}
\mathcal{E}_{s} := \exp \bigg( \frac{1}{2} \int_{t}^{s} \big( \bar{\sigma }_{x} (v) \big)^{2} dv - \int_{t}^{s} \bar{\sigma }_{x} (v) d {W}_{v} - \int_{t}^{s} \bar{b}_{x} (v) dv \bigg),
\end{equation*}
which solves the SDE $d \mathcal{E}_{s} = \mathcal{E}_{s} ( \bar{\sigma }_{x} (s) )^{2} ds - \mathcal{E}_{s} ( \bar{b}_{x} (s) ds + \bar{\sigma }_{x} (s) d {W}_{s} )$ and results in
\begin{equation*}
  \mathcal{E}_{\tau } {y}^{\varepsilon }_{\tau }
= \int_{t}^{\tau } {1}_{\{ s \in [ t, t + \varepsilon ) \}} \mathcal{E}_{s} \delta \sigma (s) \big( d {W}_{s} - \bar{\sigma }_{x} (s) ds \big)
= \int_{t}^{ \tau \wedge \varepsilon } \mathcal{E}_{s} \delta \sigma (s) \big( d {W}_{s} - \bar{\sigma }_{x} (s) ds \big).
\end{equation*}
By Doob's maximal inequality, for any $p > 1$, we obtain
\begin{align*}
      \mathbb{E} \bigg[ \sup_{ s \in [ t,T ] } | \mathcal{E}_{s} |^{p} \bigg]
& \le {e}^{ p T \sup | {b}_{x} | + p T \sup | {\sigma }_{x} |^{2} }
      \mathbb{E} \bigg[ \sup_{ s \in [ t,T ] } \Big| {e}^{ \int_{t}^{s} \bar{b}_{x} (v) dv - \int_{t}^{s} ( \bar{\sigma }_{x} (v) )^{2} dv } \mathcal{E}_{s} \Big|^{p} \bigg] \\
& \le \Big( \frac{p}{p-1} \Big)^{p} {e}^{ p T \sup | {b}_{x} | + \frac{1}{2} ( {p}^{2} + p ) T \sup | {\sigma }_{x} |^{2} }.
\end{align*}
In the same manner, we can show that
\begin{align*}
& \mathbb{E} \bigg[ \sup_{ s \in [ t,T ] } | \mathcal{E}_{s} |^{-p} \bigg]
  \le \Big( \frac{p}{p-1} \Big)^{p} {e}^{ p T \sup | \bar{b}_{x} | + \frac{1}{2} ( {p}^{2} - p ) T \sup | \bar{\sigma }_{x} |^{2} }, \quad \forall p > 1; \\
& \mathbb{E} \bigg[ \sup_{ s \in [ t,T ] } \Big| \frac{ \mathcal{E}_{s} }{ \mathcal{E}_{\tau } } \Big|^{p} \bigg]
  \le \Big( \frac{|p|}{|p|-1} \Big)^{|p|} {e}^{ |p| T \sup | \bar{b}_{x} | + \frac{1}{2} ( |p|^{2} + |p| ) T \sup | \bar{\sigma }_{x} |^{2} }, \quad \forall \tau \in [ t,T ], ~ |p| > 1.
\end{align*}
Applying the martingale representation theorem to $\{ \mathbb{E}_{s} [ \xi / \mathcal{E}_{\tau } ] \}_{ s \in [ t, \tau ] }$,
we conclude that there exists a square-integrable $\mathbb{F}$-predictable process $\{ {\gamma }_{ s, \tau } \}_{ s \in [ t, \tau ] }$
such that $\xi / \mathcal{E}_{\tau }  = \mathbb{E}_{t} [ \xi / \mathcal{E}_{\tau } ] + \int_{t}^{\tau } {\gamma }_{ s, \tau } d {W}_{s}$.
Consequently,
\begin{equation*}
  \mathbb{E}_{t} [ \xi {y}^{\varepsilon }_{\tau } ]
= \mathbb{E}_{t} \bigg[ \frac{\xi }{ \mathcal{E}_{\tau } } ( \mathcal{E}_{\tau } {y}^{\varepsilon }_{\tau } ) \bigg]
= \mathbb{E}_{t} \bigg[ \int_{t}^{ \tau \wedge \varepsilon } {\gamma }_{ s, \tau } \mathcal{E}_{s} \delta \sigma (s) ds
                      - \int_{t}^{ \tau \wedge \varepsilon } \xi \frac{ \mathcal{E}_{s} }{ \mathcal{E}_{\tau } } \delta \sigma (s) \bar{\sigma }_{x} (s) ds \bigg],
\end{equation*}
for which we have the following moment estimates with some appropriate constants $( {K}_{1}, {K}_{2}, {K}_{3}, {K}_{4} )$ independent of $\tau $:
\begin{align*}
\bigg| \mathbb{E}_{t} \bigg[ \int_{t}^{ \tau \wedge \varepsilon } {\gamma }_{ s, \tau } \mathcal{E}_{s} \delta \sigma (s) ds \bigg] \bigg|^{2}
& \le {K}_{1} \mathbb{E}_{t} \bigg[ \int_{t}^{ \tau \wedge \varepsilon } | \mathcal{E}_{s} |^{2} ds \bigg]
              \mathbb{E}_{t} \bigg[ \int_{t}^{ \tau \wedge \varepsilon } | {\gamma }_{ s, \tau } |^{2} ds \bigg] \\
& \le {K}_{2} \varepsilon \mathbb{E}_{t} \bigg[ \int_{t}^{  \tau \wedge \varepsilon } | {\gamma }_{ s, \tau } |^{2} ds \bigg], \\
  \bigg| \mathbb{E}_{t} \bigg[ \int_{t}^{ \tau \wedge \varepsilon } \xi \frac{ \mathcal{E}_{s} }{ \mathcal{E}_{\tau } } \delta \sigma (s) \bar{\sigma }_{x} (s) ds \bigg] \bigg|^{2}
& \le {K}_{3} \mathbb{E}_{t} \bigg[ \int_{t}^{ \tau \wedge \varepsilon } \Big| \frac{ \mathcal{E}_{s} }{ \mathcal{E}_{\tau } } \Big|^{2} ds \bigg]
              \mathbb{E}_{t} \bigg[ \int_{t}^{ \tau \wedge \varepsilon } {\xi }^{2} ds \bigg]
  \le {K}_{4} {\varepsilon }^{2}.
\end{align*}
Hence, $| \mathbb{E}_{t} [ \xi {y}^{\varepsilon }_{\tau } ] |^{2}
        \le 2 ( {K}_{2} \vee {K}_{4} ) ( \varepsilon \mathbb{E}_{t} [ \int_{t}^{ \tau \wedge \varepsilon } | {\gamma }_{ s, \tau } |^{2} ds ] + {\varepsilon }^{2} ) = o ( \varepsilon )$.
Furthermore, setting $\xi = \bar{b}_{x} ( \tau )$ yields
$| \mathbb{E}_{t} [ {y}^{\varepsilon }_{\tau } \bar{b}_{x} ( \tau ) ] |^{2}
\le {K}_{5} ( \varepsilon \mathbb{E}_{t} [ \int_{t}^{ \tau \wedge \varepsilon } | {\gamma }_{ s, \tau } |^{2} ds ] + {\varepsilon }^{2} )$
for some constant ${K}_{5}$ independent of $\tau $. Consequently,
\begin{equation*}
    \int_{t}^{T} | \mathbb{E}_{t} [ {y}^{\varepsilon }_{s} \bar{b}_{x} (s) ] |^{2} ds
\le {K}_{5} \bigg( \varepsilon \int_{t}^{T} \int_{t}^{ \tau \wedge \varepsilon } \mathbb{E}_{t} [ | {\gamma }_{ s, \tau } |^{2} ] ds d \tau + T {\varepsilon }^{2} \bigg)
  = o ( \varepsilon ),
\end{equation*}
where the last equality follows from the dominated convergence theorem with
\begin{equation*}
  \int_{t}^{\tau } \mathbb{E}_{t} [ | {\gamma }_{ s, \tau } |^{2} ] ds
= \mathbb{E}_{t} \bigg[ \Big| \frac{ \bar{b}_{x} ( \tau ) }{ \mathcal{E}_{\tau } } - \mathbb{E}_{t} \Big[ \frac{ \bar{b}_{x} ( \tau ) }{ \mathcal{E}_{\tau } } \Big] \Big|^{2} \bigg]
\le 4 ( \sup | {b}_{x} | )^{2} \mathbb{E}_{t} \bigg[ \sup_{ s \in [ t,T ] } | \mathcal{E}_{s} |^{-2} \bigg].
\end{equation*}
Finally, by \cref{eq:variational equations :eq} and Cauchy-Schwarz inequality, we obtain
\begin{equation*}
\sup_{ \tau \in [ t,T ] } | \mathbb{E}_{t} [ {y}^{\varepsilon }_{\tau } ] |^{2}
= \sup_{ \tau \in [ t,T ] } \bigg| \int_{t}^{\tau } \mathbb{E}_{t} [ {y}^{\varepsilon }_{s} \bar{b}_{x} (s) ] ds \bigg|^{2}
\le T \int_{t}^{T} | \mathbb{E}_{t} [ {y}^{\varepsilon }_{s} \bar{b}_{x} (s) ] |^{2} ds
  = o ( \varepsilon ),
\end{equation*}
and thus, this proof is complete.
\end{proof}

\begin{remark}
In comparison, \cite[Lemma 3.2]{Buckdahn-Djehiche-Li-2011} for proving Proposition 3.1 therein has used $\mathbb{L}^{p}$-integrability for all $p \ge 1$, which is much stronger than our condition.
\Cref{lem:estimate for square of mean} shows that 
the condition slightly stronger than square-integrability is sufficient to derive the first-order expansion of $J ( t, \bar{X}_{t}, \bar{u}^{ t, \varepsilon, \zeta } ) - J ( t, \bar{X}_{t}, \bar{u} )$ for our problem.
By \Cref{lem:estimate for square of mean}, one can further obtain the estimates
\begin{align*}
&   | \mathbb{E}_{t} [ \xi {y}^{\varepsilon }_{\tau } ] \mathbb{E}_{t} [ {z}^{\varepsilon }_{\tau } ] |
\le \frac{1}{2} ( \mathbb{E}_{t} [ \xi {y}^{\varepsilon }_{\tau } ] )^{2} + \frac{1}{2} \mathbb{E}_{t} [ | {z}^{\varepsilon }_{\tau } |^{2} ]
  = o ( \varepsilon ), \\
&   | \mathbb{E}_{t} [ {y}^{\varepsilon }_{\tau } ] \mathbb{E}_{t} [ \xi {z}^{\varepsilon }_{\tau } ] |
\le \frac{1}{2} ( \mathbb{E}_{t} [ {y}^{\varepsilon }_{\tau } ] )^{2} + \frac{1}{2} \mathbb{E}_{t} [ {\xi }^{2} ] \mathbb{E}_{t} [ | {z}^{\varepsilon }_{\tau } |^{2} ]
  = o ( \varepsilon ), \quad etc.
\end{align*}
As a result, many terms in the expansion of $J ( t, \bar{X}_{t}, \bar{u}^{ t, \varepsilon, \zeta } ) - J ( t, \bar{X}_{t}, \bar{u} )$ can be shown to be higher order infinitesimals of $\varepsilon $.
\end{remark}

For the sake of brevity, let ${M}^{u}_{j} ( t,x ) := \mathbb{E}_{t} [ ( {X}^{ t,x,u }_{T} - \mathbb{E}_{t} [ {X}^{ t,x,u }_{T} ] )^{j} ]$ for the $j$-th conditional central moment,
and $\vec{M}^{u} ( t,x ) := ( \mathbb{E}_{t} [ {X}^{ t,x,u }_{T} ], {M}^{u}_{2} ( t,x ), \ldots, {M}^{u}_{n} ( t,x ) )$.
In addition, for $f = \Psi, {\Psi }_{y}, {\Psi }_{yy}, {\Psi }_{ y {z}_{j} }, {\Psi }_{ {z}_{j} }, {\Psi }_{ {z}_{i} {z}_{j} }$,
we write $\bar{f} := f ( t, \bar{X}_{t}, \bar{X}_{T}, \vec{M}^{ \bar{u} } ( t, \bar{X}_{t} ) )$.

\begin{lemma}\label{lem:difference in objective function for spike variation}
For an arbitrarily fixed $\bar{u} \in \mathcal{U}_{0}$, any $t \in [ 0,T )$,
a sufficiently small $\varepsilon \in ( 0, T-t ]$ and an $\mathcal{F}_{t}$-measurable random variable $\zeta $ such that $\bar{u}^{ t, \varepsilon, \zeta } \in \mathcal{U}_{0}$, we have
\begin{align}
    J ( t, \bar{X}_{t}, \bar{u}^{ t, \varepsilon, \zeta } )
& = J ( t, \bar{X}_{t}, \bar{u} )
  + \mathbb{E}_{t} \Big[ \bar{\Psi }_{y} ( {y}^{\varepsilon }_{T} + {z}^{\varepsilon }_{T} ) + \frac{1}{2} \bar{\Psi }_{yy} ( {y}^{\varepsilon }_{T} )^{2} \Big]
  + \mathbb{E}_{t} [ \bar{\Psi }_{ {z}_{1} } ] \mathbb{E}_{t} [ {y}^{\varepsilon }_{T} + {z}^{\varepsilon }_{T} ] \notag \\
& \quad + \sum_{j=2}^{n} j \mathbb{E}_{t} [ \bar{\Psi }_{ {z}_{j} } ]
                           \mathbb{E}_{t} \big[ \big( ( \bar{X}_{T} - \mathbb{E}_{t} [ \bar{X}_{T} ] )^{j-1} - {M}^{ \bar{u} }_{j-1} ( t,x ) \big)
                                                ( {y}^{\varepsilon }_{T} + {z}^{\varepsilon }_{T} ) \big]
\label{eq:difference in objective function for spike variation :eq}
\\
& \quad + \sum_{j=2}^{n} \binom{j}{2} \mathbb{E}_{t} [ \bar{\Psi }_{ {z}_{j} } ] \mathbb{E}_{t} \big[ ( \bar{X}_{T} - \mathbb{E}_{t} [ \bar{X}_{T} ] )^{j-2} ( {y}^{\varepsilon }_{T} )^{2} \big]
  + o ( \varepsilon ). \notag
\end{align}
\end{lemma}

\begin{proof}
It follows from \cref{eq:difference in moments of spike variation :eq} and \Cref{lem:estimate for square of mean} that
\begin{align*}
& ( \mathbb{E}_{t} [ {X}^{\varepsilon }_{T} ] )^{j-i} \mathbb{E}_{t} [ ( {X}^{\varepsilon }_{T} )^{i} ]
- ( \mathbb{E}_{t} [ \bar{X}_{T} ] )^{j-i} \mathbb{E}_{t} [ ( \bar{X}_{T} )^{i} ] \\
& = (j-i) {1}_{\{ i \le j-1 \}} ( \mathbb{E}_{t} [ \bar{X}_{T} ] )^{j-i-1}
    \mathbb{E}_{t} [ ( \bar{X}_{T} )^{i} ] \mathbb{E}_{t} [ {y}^{\varepsilon }_{T} + {z}^{\varepsilon }_{T} ] \\
& + {1}_{\{ i \ge 1 \}} i ( \mathbb{E}_{t} [ \bar{X}_{T} ] )^{j-i}
    \mathbb{E}_{t} [ ( \bar{X}_{T} )^{i-1} ( {y}^{\varepsilon }_{T} + {z}^{\varepsilon }_{T} ) ] \\
& + {1}_{\{ i \ge 2 \}} \binom{i}{2} ( \mathbb{E}_{t} [ \bar{X}_{T} ] )^{j-i}
    \mathbb{E}_{t} [ ( \bar{X}_{T} )^{i-2} ( {y}^{\varepsilon }_{T} )^{2} ]
  + o ( \varepsilon ),
\end{align*}
which leads to
\begin{align}
  {M}^{ \bar{u}^{ t, \varepsilon, \zeta } }_{j} ( t, \bar{X}_{t} )
& = {M}^{\bar{u} }_{j} ( t, \bar{X}_{t} )
  + j \mathbb{E}_{t} \Big[ \big( ( \bar{X}_{T} - \mathbb{E}_{t} [ \bar{X}_{T} ] )^{j-1} - {M}_{j-1}^{\bar{u}} ( t, \bar{X}_{t} ) \big)
                           ( {y}^{\varepsilon }_{T} + {z}^{\varepsilon }_{T} ) \Big] \notag \\
& \quad + \binom{j}{2} \mathbb{E}_{t} \big[ ( \bar{X}_{T} - \mathbb{E}_{t} [ \bar{X}_{T} ] )^{j-2} ( {y}^{\varepsilon }_{T} )^{2} \big]
        + o ( \varepsilon ), \quad j \ge 2.
\label{eq:difference in central moments for spike variation :eq}
\end{align}
On the one hand, in the same line of proof as for \cite[(4.31), p. 128]{Yong-Zhou-1999}, one can obtain
\begin{align*}
& \mathbb{E}_{t} \big[ \Psi \big( t, \bar{X}_{t}, {X}^{\varepsilon }_{T}, \vec{M}^{ \bar{u}^{ t, \varepsilon, \zeta } } ( t, \bar{X}_{t} ) \big) \big]
- \mathbb{E}_{t} \big[ \Psi \big( t, \bar{X}_{t}, \bar{X}_{T}, \vec{M}^{ \bar{u}^{ t, \varepsilon, \zeta } } ( t, \bar{X}_{t} ) \big) \big] \\
& = \mathbb{E}_{t} \bigg[ {\Psi }_{y} \big( t, \bar{X}_{t}, \bar{X}_{T}, \vec{M}^{ \bar{u}^{ t, \varepsilon, \zeta } } ( t, \bar{X}_{t} ) \big) ( {y}^{\varepsilon }_{T} + {z}^{\varepsilon }_{T} ) \bigg] \\
& + \frac{1}{2} \mathbb{E}_{t} \bigg[ {\Psi }_{yy} \big( t, \bar{X}_{t}, \bar{X}_{T}, \vec{M}^{ \bar{u}^{ t, \varepsilon, \zeta } } ( t, \bar{X}_{t} ) \big) ( {y}^{\varepsilon }_{T} )^{2} \bigg]
  + o ( \varepsilon ).
\end{align*}
By \cref{eq:moment estimates :eq}, \cref{eq:difference in central moments for spike variation :eq}, \Cref{lem:estimate for square of mean}
and the Lipschitz continuity of $( \Phi, {\Phi }_{y} )$, we obtain
\begin{align*}
& \mathbb{E}_{t} \big[ {\Psi }_{y} \big( t, \bar{X}_{t}, \bar{X}_{T}, \vec{M}^{ \bar{u}^{ t, \varepsilon, \zeta } } ( t, \bar{X}_{t} ) \big) {y}^{\varepsilon }_{T} \big]
- \mathbb{E}_{t} [ \bar{\Psi }_{y} {y}^{\varepsilon }_{T} \big]
= o ( \varepsilon ), \\
& \mathbb{E}_{t} \big[ {\Psi }_{y} \big( t, \bar{X}_{t}, \bar{X}_{T}, \vec{M}^{ \bar{u}^{ t, \varepsilon, \zeta } } ( t, \bar{X}_{t} ) \big) {z}^{\varepsilon }_{T} \big]
- \mathbb{E}_{t} [ \bar{\Psi }_{y} {z}^{\varepsilon }_{T} ]
= o ( \varepsilon ), \\
& \mathbb{E}_{t} \big[ {\Psi }_{yy} \big( t, \bar{X}_{t}, \bar{X}_{T}, \vec{M}^{ \bar{u}^{ t, \varepsilon, \zeta } } ( t, \bar{X}_{t} ) \big) ( {y}^{\varepsilon }_{T} )^{2} \big]
- \mathbb{E}_{t} [ \bar{\Psi }_{yy} ( {y}^{\varepsilon }_{T} )^{2} ]
= o ( \varepsilon ).
\end{align*}
On the other hand, by Taylor's expansion with the use of \cref{eq:moment estimates :eq}, \cref{eq:difference in central moments for spike variation :eq} and \Cref{lem:estimate for square of mean}, we can show that
\begin{align*}
& \mathbb{E}_{t} \big[ \Psi \big( t, \bar{X}_{t}, \bar{X}_{T}, \vec{M}^{ \bar{u}^{ t, \varepsilon, \zeta } } ( t, \bar{X}_{t} ) \big) \big] - \mathbb{E}_{t} [ \bar{\Psi } ]
  - \mathbb{E}_{t} [ \bar{\Psi }_{ {z}_{1} } ] \mathbb{E}_{t} [ {y}^{\varepsilon }_{T} + {z}^{\varepsilon }_{T} ] \\
& = \sum_{j=2}^{n} j \mathbb{E}_{t} [ \bar{\Psi }_{ {z}_{j} } ]
                     \mathbb{E}_{t} \big[ \big( ( \bar{X}_{T} - \mathbb{E}_{t} [ \bar{X}_{T} ] )^{j-1} - {M}^{ \bar{u} }_{j-1} ( t, \bar{X}_{t} ) \big) ( {y}^{\varepsilon }_{T} + {z}^{\varepsilon }_{T} ) \big] \\
& \quad + \sum_{j=2}^{n} \frac{ j (j-1) }{2}
                         \mathbb{E}_{t} [ \bar{\Psi }_{ {z}_{j} } ]
                         \mathbb{E}_{t} \big[ ( \bar{X}_{T} - \mathbb{E}_{t} [ \bar{X}_{T} ] )^{j-2} ( {y}^{\varepsilon }_{T} )^{2} \big]
        + o ( \varepsilon ).
\end{align*}
Therefore, plugging the above-mentioned expansions into the decomposition
\begin{align*}
& J ( t, \bar{X}_{t}, \bar{u}^{ t, \varepsilon, \zeta } ) - J ( t, \bar{X}_{t}, \bar{u} ) \\
& = \Big( \mathbb{E}_{t} \big[ \Psi \big( t, \bar{X}_{t}, {X}^{\varepsilon }_{T}, \vec{M}^{ \bar{u}^{ t, \varepsilon, \zeta } } ( t, \bar{X}_{t} ) \big) \big]
        - \mathbb{E}_{t} \big[ \Psi \big( t, \bar{X}_{t}, \bar{X}_{T}, \vec{M}^{ \bar{u}^{ t, \varepsilon, \zeta } } ( t, \bar{X}_{t} ) \big) \big] \Big) \\
& \quad + \Big( \mathbb{E}_{t} \big[ \Psi \big( t, \bar{X}_{t}, \bar{X}_{T}, \vec{M}^{ \bar{u}^{ t, \varepsilon, \zeta } } ( t, \bar{X}_{t} ) \big) \big]
              - \mathbb{E}_{t} [ \bar{\Psi } ] \Big)
\end{align*}
immediately yields the desired result.
\end{proof}

Now we are able to state the maximum principles for our control problems.

\begin{theorem}\label{thm:sufficient maximum principle}
Suppose that corresponding to $\bar{u} \in \mathcal{U}_{0}$,
the following FBSDE admits a square-integrable solution $( \bar{X}, {Y}^{t}, \mathcal{Y}^{t}, {Z}^{t}, \mathcal{Z}^{t} )$ for a.e. $t \in [ 0,T )$:
\begin{equation}
\left\{ \begin{aligned}
d \bar{X}_{s} & = \bar{b} (s) ds + \bar{\sigma } (s) d {W}_{s}, ~ \forall s \in [ 0,T ], \quad
  \bar{X}_{0}   = {x}_{0}; \\
d {Y}^{t}_{s} & = - \big( {Y}^{t}_{s} \bar{b}_{x} (s) + \mathcal{Y}^{t}_{s} \bar{\sigma }_{x} (s) \big) ds
                  + \mathcal{Y}^{t}_{s} d {W}_{s}, ~ \forall s \in [ t,T ], \\
  {Y}^{t}_{T} & = \bar{\Psi }_{y}
                + \mathbb{E}_{t} [ \bar{\Psi }_{ {z}_{1} } ]
                + \sum_{j=2}^{n} j \mathbb{E}_{t} [ \bar{\Psi }_{ {z}_{j} } ] \big( ( \bar{X}_{T} - \mathbb{E}_{t} [ \bar{X}_{T} ] )^{j-1} - {M}^{ \bar{u} }_{j-1} ( t, \bar{X}_{t} ) \big); \\
d {Z}^{t}_{s} & = - \big( 2 {Z}^{t}_{s} \bar{b}_{x} (s) + {Z}^{t}_{s} | \bar{\sigma }_{x} (s) |^{2} + 2 \mathcal{Z}^{t}_{s} \bar{\sigma }_{x} (s)
                        + {Y}^{t}_{s} \bar{b}_{xx} (s) + \mathcal{Y}^{t}_{s} \bar{\sigma }_{xx} (s) \big) ds \\
              & \quad + \mathcal{Z}^{t}_{s} d {W}_{s}, ~ \forall s \in [ t,T ], \\
  {Z}^{t}_{T} & = \bar{\Psi }_{yy}
                + \sum_{j=2}^{n} j (j-1) \mathbb{E}_{t} [ \bar{\Psi }_{ {z}_{j} } ] ( \bar{X}_{T} - \mathbb{E}_{t} [ \bar{X}_{T} ] )^{j-2},
\end{aligned} \right. 
\label{eq:FBSDE :eq}
\end{equation}
such that
\begin{equation*}
\limsup_{ s \downarrow t } \mathbb{E}_{t} \bigg[ \frac{1}{2} {Z}^{t}_{s} \big( | \sigma ( s, \bar{X}_{s}, \bar{u}_{s} + \zeta ) |^{2} - | \bar{\sigma } (s) |^{2} \big)
                                               + {Y}^{t}_{s} \delta b(s) + \big( \mathcal{Y}^{t}_{s} - {Z}^{t}_{s} \bar{\sigma } (s) \big) \delta \sigma (s) \bigg] \le 0
\end{equation*}
for any $\mathcal{F}_{t}$-measurable random variable $\zeta $ satisfying $\bar{u}^{ t, \varepsilon, \zeta } \in \mathcal{U}_{0}$.
Then, $\bar{u}$ is an ONEC.
\end{theorem}

\begin{proof}
According to the FBSDE \cref{eq:FBSDE :eq}, \cref{eq:difference in objective function for spike variation :eq} can be re-expressed as
\begin{equation*}
  J ( t, \bar{X}_{t}, \bar{u}^{ t, \varepsilon, \zeta } ) - J ( t, \bar{X}_{t}, \bar{u} )
= \mathbb{E}_{t} [ {Y}^{t}_{T} ( {y}^{\varepsilon }_{T} + {z}^{\varepsilon }_{T} ) ]
+ \frac{1}{2} \mathbb{E}_{t} [ {Z}^{t}_{T} ( {y}^{\varepsilon }_{T} )^{2} ] + o ( \varepsilon ).
\end{equation*}
In the same line of duality analysis as in \cite[Section 3.4.3, pp. 134-137]{Yong-Zhou-1999}, we obtain
\begin{equation}
\begin{aligned}
& J ( t, \bar{X}_{t}, \bar{u}^{ t, \varepsilon, \zeta } ) - J ( t, \bar{X}_{t}, \bar{u} ) \\
& = \mathbb{E}_{t} \bigg[ \int_{t}^{ t + \varepsilon } \Big( {Y}^{t}_{s} \delta b (s) + \mathcal{Y}^{t}_{s} \delta \sigma (s) + \frac{1}{2} {Z}^{t}_{s} | \delta \sigma (s) |^{2} \Big) ds \bigg]
  + o ( \varepsilon ) \\
& = \mathbb{E}_{t} \bigg[ \int_{t}^{ t + \varepsilon } \Big( \frac{1}{2} {Z}^{t}_{s} \big( | \sigma ( s, \bar{X}_{s}, \bar{u}_{s} + \zeta ) |^{2} - | \bar{\sigma } (s) |^{2} \big) \\
& \qquad \qquad \qquad                                     + {Y}^{t}_{s} \delta b(s) + \big( \mathcal{Y}^{t}_{s} - {Z}^{t}_{s} \bar{\sigma } (s) \big) \delta \sigma (s) \Big) ds \bigg]
  + o ( \varepsilon )
\le o ( \varepsilon ),
\end{aligned}
\label{eq:difference in objective function of spike variation :eq}
\end{equation}
which implies that $\bar{u}$ is an ONEC.
\end{proof}

\begin{theorem}\label{thm:sufficient maximum principle for linear control problem}
Suppose that $b ( t,x,u )$ and $\sigma ( t,x,u )$ are both affine in $u$.
Then, $\bar{u} \in \mathcal{U}_{0}$ is an ONEC, if for a.e. $t \in [ 0,T )$, the FBSDE \cref{eq:FBSDE :eq} admits a square-integrable solution $( \bar{X}, {Y}^{t}, \mathcal{Y}^{t}, {Z}^{t}, \mathcal{Z}^{t} )$ such that
\begin{equation}
\lim_{s \downarrow t} \mathbb{E}_{t} [ {Y}^{t}_{s} {b}_{u} ( s, \bar{X}_{s}, u ) + \mathcal{Y}^{t}_{s} {\sigma }_{u} ( s, \bar{X}_{s}, u ) ] = 0, \quad
\lim_{ s \downarrow t } \mathbb{E}_{t} [ {Z}^{t}_{s} ] \le 0, \quad \mathbb{P}-a.s.
\label{eq:optimality condition for ONEC of linear control problem :eq}
\end{equation}
\end{theorem}

\begin{proof}
Given that $\lim_{ s \downarrow t } \mathbb{E}_{t} [ {Z}^{t}_{s} ] \le 0$,
from the first equality in \cref{eq:difference in objective function of spike variation :eq} we have
\begin{equation*}
\lim_{ \varepsilon \downarrow 0 } \frac{1}{\varepsilon } \big( J ( t, \bar{X}_{t}, \bar{u}^{ t, \varepsilon, \zeta } ) - J ( t, \bar{X}_{t}, \bar{u} ) \big)
\le \lim_{ \varepsilon \downarrow 0 } \frac{\zeta }{\varepsilon } \int_{t}^{ t + \varepsilon } \mathbb{E}_{t} [ {Y}^{t}_{s} {b}_{u} ( s, \bar{X}_{s}, u ) + \mathcal{Y}^{t}_{s} {\sigma }_{u} ( s, \bar{X}_{s}, u ) ] ds.
\end{equation*}
Due to the square-integrability of $( {Y}^{t}, \mathcal{Y}^{t} )$ and the Lipschitz continuity of $b ( t,x, \cdot )$ and $\sigma ( t,x, \cdot )$,
we obtain $\mathbb{E}_{t} \int_{t}^{T} | {Y}^{t}_{s} {b}_{u} ( s, \bar{X}_{s}, u ) + \mathcal{Y}^{t}_{s} {\sigma }_{u} ( s, \bar{X}_{s}, u ) | ds < \infty $.
Hence,
\begin{equation*}
    \lim_{ \varepsilon \downarrow 0 } \frac{1}{\varepsilon } \big( J ( t, \bar{X}_{t}, \bar{u}^{ t, \varepsilon, \zeta } ) - J ( t, \bar{X}_{t}, \bar{u} ) \big)
\le \zeta \lim_{s \downarrow t} \mathbb{E}_{t} [ {Y}^{t}_{s} {b}_{u} ( s, \bar{X}_{s}, u ) + \mathcal{Y}^{t}_{s} {\sigma }_{u} ( s, \bar{X}_{s}, u ) ]
  = 0,
\end{equation*}
which implies that $\bar{u}$ is an ONEC.
\end{proof}

\begin{example}\label{ex:non-linear law-dependent preferences}
There is a recent work \cite{Liang-Xia-Yuan-2023} that characterized open-loop equilibria for dynamic portfolio problem with a fairly general objective function $g ( \mathbb{P}^{t}_{ {X}_{T} } )$ named non-linear law-dependent preference, 
where $\mathbb{P}^{t}_{ {X}_{T} }$ denotes the regular conditional distribution of ${X}_{T}$ given $\mathcal{F}_{t}$.
To make a comparison with \cite{Liang-Xia-Yuan-2023}, we consider $\Psi ( \vec{z} ) = g ( {p}^{t} ( \cdot, \vec{z} ) )$, 
where ${p}^{t} ( \cdot, \vec{z} )$ is a $\vec{z}$-indexed regular conditional law given $\mathcal{F}_{t}$ under $\mathbb{P}$.
This setting captures some non-linear law-dependent preferences as in \cite{Liang-Xia-Yuan-2023}, 
as ${p}^{t} ( \cdot, \vec{M}^{u} ( t,x ) )$ could be the conditional law of ${X}^{t,x,u}_{T}$ characterized by its conditional expectation and higher-order central moments.
Let $\mathcal{B} ( \mathbb{R} )$ be the Borel $\sigma $-algebra generated by all open intervals in $\mathbb{R}$,
$\mathcal{P}$ be the Banach space of all finite signed measures on $( \mathbb{R}, \mathcal{B} ( \mathbb{R} ) )$,
the non-empty convex set $\mathcal{P}_{0} \subseteq \mathcal{P}$ be the domain of the fairly general functional $g$,
and $\mathcal{L} ( \mathcal{S}_{1}, \mathcal{S}_{2} )$ be the collection of all continuous linear functions from $\mathcal{S}_{1}$ to $\mathcal{S}_{2}$.
In addition to the basic assumptions for $\Psi $ in \Cref{sec:Model formulation}, we shall introduces the following assumptions that can be viewed as continuous G\^ateaux differentiability:
\begin{enumerate}
\item There exists $\nabla g: \mathcal{P}_{0} \times \mathbb{R} \to \mathbb{R}$ such that for any $p,q \in \mathcal{P}_{0}$,
      \begin{equation*}
      \lim_{\varepsilon \downarrow 0} \frac{ g ( p + \varepsilon (q-p) ) - g ( p ) }{\varepsilon } =: dg ( p, q-p ) = \int_{\mathbb{R}} \nabla g ( p, x ) \big( q(dx) - p(dx) \big),
      \end{equation*}
      $\int_{\mathbb{R}} | \nabla g ( p,x ) | q(dx) < \infty $, and $p \mapsto dg ( p, \cdot ) \in \mathcal{L} ( \mathcal{P}, \mathbb{R} )$ is continuous.
\item For any $p \in \mathcal{P}_{0}$, there exists ${p}_{\vec{z}}: \mathcal{B} ( \mathbb{R} ) \times \mathbb{R}^{n} \to \mathbb{R}^{n}$ such that for any $\vec{z}, \vec{h} \in \mathbb{R}^{n}$,
      \begin{equation*}
      \lim_{\varepsilon \downarrow 0} \frac{ p ( B, \vec{z} + \varepsilon \vec{h} ) - p ( B, \vec{z} ) }{\varepsilon } = \langle {p}_{\vec{z}} ( B, \vec{z} ), \vec{h} \rangle, \quad \forall B \in \mathcal{B} ( \mathbb{R} ),
      \end{equation*}
      and $\vec{z} \mapsto \int_{\mathbb{R}} \nabla g ( q, x ) \langle {p}_{\vec{z}} ( dx, \vec{z} ), \cdot \rangle \in \mathcal{L} ( \mathbb{R}^{n}, \mathbb{R}^{n} )$ is continuous for any $q \in \mathcal{P}_{0}$.
\end{enumerate} 
Notably, $\nabla g$ is not unique due to $\int_{\mathbb{R}} ( q (dx) - p (dx) ) = 0$ for probability measures $( p,q )$.
By the chain rule, ${\Psi }_{ \vec{z} } ( \vec{z} ) = \int_{\mathbb{R}} \nabla g ( {p}^{t} ( \cdot, \vec{z} ), x ) {p}^{t}_{ \vec{z} } (dx)$ follows.
Then, according to \Cref{thm:sufficient maximum principle,thm:sufficient maximum principle for linear control problem}, the ONEC for $\Psi ( \vec{M}^{u} ( t,x ) ) = g ( {p}^{t} ( \cdot, \vec{M}^{u} ( t,x ) ) )$ can be characterized by \cref{eq:FBSDE :eq} with 
\begin{equation}
\left\{ \begin{aligned}
& {Y}^{t}_{T} = \int_{\mathbb{R} } \nabla g \Big( {p}^{t} \big( \cdot, \vec{M}^{\bar{u}} ( t,x ) \big), x \Big) \big\langle {p}^{t}_{\vec{z}} \big( dx, \vec{M}^{\bar{u}} ( t,x ) \big), \vec{Q} \big( \bar{X}_{T}, \vec{M}^{\bar{u}} ( t,x ) \big) \big\rangle, \\
& {Z}^{t}_{T} = \int_{\mathbb{R} } \nabla g \Big( {p}^{t} \big( \cdot, \vec{M}^{\bar{u}} ( t,x ) \big), x \Big) \big\langle {p}^{t}_{\vec{z}} \big( dx, \vec{M}^{\bar{u}} ( t,x ) \big), \vec{Q}_{x} \big( \bar{X}_{T}, \vec{M}^{\bar{u}} ( t,x ) \big) \big\rangle,
\end{aligned} \right.
\label{eq:FBSDE terminal condition for law-dependent preference :eq}
\end{equation}
where $\vec{Q} ( x, \vec{z} ) := ( {Q}_{1} ( x, \vec{z} ), \ldots, {Q}_{n} ( x, \vec{z} ) )$
for ${Q}_{1} ( x, \vec{z} ) := 1$, ${Q}_{2} ( x, \vec{z} ) := 2 ( x - {z}_{1} )$ and other ${Q}_{j} ( x, \vec{z} ) := j ( x - {z}_{1} )^{j-1} - j {z}_{j-1}$,
and $\vec{Q}_{x} ( x, \vec{z} ) := {\partial }_{x} \vec{Q} ( x, \vec{z} )$.
On the other hand, we write $\vec{M}^{t}_{X} := ( \mathbb{E}_{t} [X], \mathbb{E}_{t} [ ( X - \mathbb{E}_{t} [X] )^{2} ], \ldots, \mathbb{E}_{t} [ ( X - \mathbb{E}_{t} [X] )^{n} ] )$ analogous to $\vec{M}^{u} ( t, {X}_{t} )$,
and then obtain
\begin{equation*}
\lim_{\varepsilon \to 0} \frac{1}{\varepsilon } \bigg\| \vec{M}^{t}_{ X + \varepsilon \Delta } - \vec{M}^{t}_{X} - \mathbb{E}_{t} \bigg[ \int_{X}^{ X + \varepsilon \Delta } \vec{Q} ( v, \vec{M}^{t}_{X} ) \rangle dv \bigg] \bigg\| = 0,
\end{equation*}
where $\| \cdot \|$ is the Euclidean norm, from
\begin{equation*}
  \mathbb{E}_{t} [ ( X + \varepsilon \Delta - \mathbb{E}_{t} [ X + \varepsilon \Delta ] )^{j} ] - \mathbb{E}_{t} [ ( X - \mathbb{E}_{t} [X] )^{j} ]
= \mathbb{E}_{t} \bigg[ \int_{X}^{ X + \varepsilon \Delta } {Q}_{j} ( v, \vec{M}_{X} ) dv \bigg] + o ( \varepsilon ).
\end{equation*}
Consequently, for an arbitrarily fixed $c \in \mathbb{R}$,
\begin{align*}
&   g \big( {p}^{t} ( \cdot, \vec{M}^{t}_{ X + \varepsilon \Delta } ) \big) - g \big( {p}^{t} ( \cdot, \vec{M}^{t}_{X} ) \big) \\
& = \mathbb{E}_{t} \bigg[ \int_{X}^{ X + \varepsilon \Delta } dv \int_{\mathbb{R}} \nabla g \big( {p}^{t} ( \cdot, \vec{M}^{t}_{X} ), x \big) \langle {p}^{t}_{\vec{z}} ( dx, \vec{M}^{t}_{X} ), \vec{Q} ( v, \vec{M}^{t}_{X} ) \rangle \bigg] + o ( \varepsilon ) \\
& = \int_{\mathbb{R}} \bigg( \int_{c}^{s} \int_{\mathbb{R}} \nabla g \big( {p}^{t} ( \cdot, \vec{M}^{t}_{X} ), x \big) \langle {p}^{t}_{\vec{z}} ( dx, \vec{M}^{t}_{X} ), \vec{Q} ( v, \vec{M}^{t}_{X} ) \rangle dv \bigg)
                      \big( \mathbb{P}^{t}_{ X + \varepsilon \Delta } (ds) - \mathbb{P}^{t}_{X} (ds) \big) \\
& \quad + o ( \varepsilon )
\end{align*}
where $\mathbb{P}^{t}_{X}$ denotes the regular conditional law of $X$ given $\mathcal{F}_{t}$ under $\mathbb{P}$. 
This implies that if fixing ${p}^{t}$ and writing $G ( \mathbb{P}^{t}_{X} ) = g ( {p}^{t} ( \cdot, \vec{M}^{t}_{X} ) )$, then one can choose
\begin{equation*}
\nabla G ( \mathbb{P}^{t}_{X}, x ) = \int_{c}^{x} \int_{\mathbb{R}} \nabla g \big( {p}^{t} ( \cdot, \vec{M}^{t}_{X} ), s \big) \langle {p}^{t}_{\vec{z}} ( ds, \vec{M}^{t}_{X} ), \vec{Q} ( v, \vec{M}^{t}_{X} ) \rangle dv, \quad \forall c \in \mathbb{R},
\end{equation*}
so that \cref{eq:FBSDE terminal condition for law-dependent preference :eq} can be re-expressed as 
\begin{equation*}
{Y}^{t}_{T} = {\partial }_{x} \nabla G ( \mathbb{P}^{t}_{ \bar{X}_{T} }, x ) |_{ x = \bar{X}_{T} }, \quad 
{Z}^{t}_{T} = {\partial }_{x}^{2} \nabla G ( \mathbb{P}^{t}_{ \bar{X}_{T} }, x ) |_{ x = \bar{X}_{T} }.
\end{equation*}
Furthermore, if ${p}^{t} ( \cdot, \vec{M}^{t}_{\bar{X}_{T}} ) = \mathbb{P}^{t}_{\bar{X}_{T}}$, then $G \equiv g$, and the corresponding verification result in \cite{Liang-Xia-Yuan-2023} gets reproduced to some extent.
Notably, taking advantage of the second-order adjoint equation for ${Z}^{t}$, we do not need to impose the concavity condition on $\nabla g ( \mathbb{P}^{t}_{\bar{X}_{T}}, \cdot )$ as in \cite[Assumption 3.1]{Liang-Xia-Yuan-2023}.
For instance, if the functions $( b, \sigma )$ are affine in each of $( x,u )$ as in \cite{Liang-Xia-Yuan-2023},
then it follows from \cref{eq:FBSDE :eq} that 
\begin{equation*}
{Z}^{t}_{s} = \mathbb{E}_{s} \Big[ {Z}^{t}_{T} {e}^{ 2 \int_{s}^{T} \bar{\sigma }_{x} (v) d {W}_{v} - \int_{s}^{T} | \bar{\sigma }_{x} (v) |^{2} dv + 2 \int_{s}^{T} \bar{b}_{x} (v) dv } \Big],
\end{equation*}
which automatically satisfies that $\lim_{ s \downarrow t } \mathbb{E}_{t} [ {Z}^{t}_{s} ] \le 0$ under the concavity assumption of $\nabla g ( \mathbb{P}^{t}_{\bar{X}_{T}}, \cdot )$.  
\end{example}

\section{Particular cases with linear controlled SDE}
\label{sec:Particular case: linear controlled SDE}

Let $\mathcal{R} = \mathbb{R}$, $b ( t,x,u ) = {A}_{t} x + {B}_{t} u + {C}_{t}$ and $\sigma ( t,x,u ) = {D}_{t} u + {F}_{t}$,
where $A, B, C, D, F: [ 0,T ] \to \mathbb{R}$ are deterministic continuous functions with $| {D}_{t} | > 0$ for every $t$, and the constant $\kappa \ge 0$.
Suppose that $\Psi ( t, x, {X}_{T}, \vec{z} )$ is independent of $x$ and ${X}_{T}$ and affine in ${z}_{1}$, and write it as $\kappa {z}_{1} + \psi ( t, {z}_{2}, \ldots, {z}_{n} )$.
That is, the objective function is
\begin{equation}
J ( t,x,u ) = \kappa \mathbb{E}_{t} [ {X}^{t,x,u}_{T} ] + \psi \big( t, {M}^{u}_{2} ( t,x ), \ldots, {M}^{u}_{n} ( t,x ) \big),
\label{eq:objective functional particular :eq}
\end{equation}
and the dynamics of ${X}^{t,x,u}$ evolve as the linear controlled SDE:
\begin{equation}
d {X}^{t,x,u}_{s} = ( {A}_{s} {X}^{t,x,u}_{s} + {B}_{s} {u}_{s} + {C}_{s} ) ds + ( {D}_{s} {u}_{s} + {F}_{s} ) d {W}_{s}, \quad
  {X}^{t,x,u}_{t} = x.
\label{eq:linear controlled SDE :eq}
\end{equation}
In addition, for the sake of brevity, we let $\vec{M}_{-1}^{u} ( t,x ) := ( {M}^{u}_{2} ( t,x ), \ldots, {M}^{u}_{n} ( t,x ) )$,
${\alpha }_{0} ( \cdot ) = 1$, ${\alpha }_{1} ( \cdot ) = 0$,
$\vec{\alpha } (y) := ( {\alpha }_{2} (y), \ldots, {\alpha }_{n} (y) )$ for ${\alpha }_{j} (y) := {1}_{ \{ j \in 2 \mathbb{N} \} } (j-1)!! {y}^{j/2}$,
${y}^{\beta }_{t} := \int_{t}^{T} | {D}_{s} {\beta }_{s} |^{2} ds$,
\begin{equation*}
{\theta }_{t} := \int_{t}^{T} \Big| \frac{ {B}_{s} }{ {D}_{s} } \Big|^{2} ds, \quad and \quad
\Theta ( t,x ) := x {e}^{ \int_{t}^{T} {A}_{v} dv } + \int_{t}^{T} {e}^{ \int_{s}^{T} {A}_{v} dv } \Big( {C}_{s} - \frac{ {B}_{s} }{ {D}_{s} } {F}_{s} \Big) ds.
\end{equation*}
Here $(j-1)!!$ denotes the double factorial of $j-1$, or namely, the product of all the positive integers up to $j-1$ that have the same parity (odd or even) as $j-1$.

\begin{remark}
The objective function \cref{eq:objective functional particular :eq} arises from the theory of efficient frontier. See also \cite{Li-Ng-2000,Zhou-Li-2000}.
Generally, the investors try to make trade-offs between the mean and the combination of higher-order central moments, which respectively stand for reward and risk.
To see an economic insight, one can let the state $X$ be the household/individual wealth,
and $u$ be the investment amount on share, with the interest rate $A$, the volatility rate $D$, the market price of risk $B / D$, the income rate $C$ and the white noise coefficient $F$.
Notably, in the coming \Cref{thm:CNEC particular}, the CNEC is not only independent of the wealth level $X$, but also independent of the income rate $C$.
\end{remark}

\subsection{Closed-loop Nash equilibrium control}

According to \cref{eq:objective functional rearrangement for our problem :eq},
\begin{equation}
\Phi ( t, \cdot, \cdot, \vec{z} )
\equiv \phi ( t, \vec{z} )
 := \kappa {z}_{1} + \psi \bigg( t, {z}_{2} - {z}_{1}^{2}, \ldots, \sum_{j=0}^{n} \binom{n}{j} (-1)^{j} {z}_{1}^{j} {z}_{n-j} \bigg).
\label{eq:objective functional rearrangement specific :eq}
\end{equation}
Then, we conclude that ${U}^{ s, \cdot, \vec{z} } ( \cdot, \cdot ) \equiv \phi ( s, \vec{z} )$
and $\vec{\lambda }^{s,y} ( \cdot, \cdot ) \equiv {\phi }_{ \vec{z} } ( s, \vec{m}^{ \tilde{u} } ( s,y ) )$ for \Cref{thm:verification theorem}.
Therefore, to determine a CNEC, it suffices to find a solution for
\begin{equation}
\left\{ \begin{aligned}
& \mathcal{D}_{ \tilde{u} ( t,x ) } \vec{m}  ( t,x ) = 0, \quad s.t. \quad \vec{m} ( T,x ) = ( x, {x}^{2}, \ldots, {x}^{n} ); \\
& \tilde{u} ( t,x ) = - \frac{ {B}_{t} }{ | {D}_{t} |^{2} }
                        \frac{ \big\langle {\phi }_{ \vec{z} } \big( t, \vec{m} ( t,x ) \big), \vec{m}_{x} ( t,x ) \big\rangle }
                             { \big\langle {\phi }_{ \vec{z} } \big( t, \vec{m} ( t,x ) \big), \vec{m}_{xx} ( t,x ) \big\rangle }
                      - \frac{ {F}_{t} }{ {D}_{t} }; \\
& 0 > \big\langle {\phi }_{ \vec{z} } \big( t, \vec{m} ( t,x ) \big), \vec{m}_{xx} ( t,x ) \big\rangle.
\end{aligned} \right.
\label{eq:EHJBS specific :eq}
\end{equation}

\begin{theorem}\label{thm:CNEC particular}
Let $\{ {\beta }_{t} \}_{ t \in [ 0,T ] }$ fulfill the integral equation
\begin{equation}
\frac{ \kappa {B}_{t} }{ | {D}_{t} |^{2} } + {\beta }_{t} \sum_{ 1 \le j \le \frac{n}{2} } 2j (2j-1) {\alpha }_{2j-2} ( {y}^{\beta }_{t} ) {\psi }_{ {z}_{2j} } \big( t, \vec{\alpha } ( {y}^{\beta }_{t} ) \big) = 0,
\label{eq:integral equation for CNEC particular :eq}
\end{equation}
and satisfy the condition
\begin{equation}
\sum_{ 1 \le j \le \frac{n}{2} } j (2j-1) {\alpha }_{2j-2} ( {y}^{\beta }_{t} ) {\psi }_{ {z}_{2j} } \big( t, \vec{\alpha } ( {y}^{\beta }_{t} ) \big) < 0.
\label{eq:concavity condition for CNEC particular :eq}
\end{equation}
Then, $( \tilde{u}, \vec{m}^{ \tilde{u} } )$ determined by
\begin{equation}
\tilde{u} ( t,x ) = {\beta }_{t} {e}^{ - \int_{t}^{T} {A}_{v} dv } - \frac{ {F}_{t} }{ {D}_{t} }
\label{eq:CNEC particular :eq}
\end{equation}
solves the PDE system \cref{eq:EHJBS specific :eq},
and hence, $\tilde{u}$ is a CNEC for the control problem with the objective function \cref{eq:objective functional particular :eq} and the controlled SDE \cref{eq:linear controlled SDE :eq}.
\end{theorem}

\begin{proof}
Plugging \cref{eq:CNEC particular :eq} into the SDE in \cref{eq:linear controlled SDE :eq} yields
\begin{equation*}
d {X}^{ t,x, \tilde{u} }_{s} = \Big( {A}_{s} {X}^{ t,x, \tilde{u} }_{s} + {B}_{s} {\beta }_{s} {e}^{ - \int_{s}^{T} {A}_{v} dv } + {C}_{s} - \frac{ {B}_{s} }{ {D}_{s} } {F}_{s} \Big) ds
                             + {D}_{s} {\beta }_{s} {e}^{ - \int_{s}^{T} {A}_{v} dv } d {W}_{s},
\end{equation*}
which leads to ${X}^{ t, x, \tilde{u} }_{T} = \Theta ( t,x ) + \int_{t}^{T} {B}_{s} {\beta }_{s} ds + \int_{t}^{T} {D}_{s} {\beta }_{s} d {W}_{s}$.
Since ${\partial }_{x} {X}^{ t,x, \tilde{u} }_{T} = \exp ( \int_{t}^{T} {A}_{v} dv )$ is deterministic, then
\begin{equation}
\left\{ \begin{aligned}
    {\partial }_{x} {m}^{ \tilde{u} }_{j} ( t,x )
& = j \mathbb{E}_{t} [ ( {X}^{ t,x, \tilde{u} }_{T} )^{j-1} ] {\partial }_{x} {X}^{ t,x, \tilde{u} }_{T}
  = j {m}^{ \tilde{u} }_{j-1} ( t,x ) {e}^{ \int_{t}^{T} {A}_{v} dv }, \\
    {\partial }_{x}^{2} {m}^{ \tilde{u} }_{j} ( t,x )
& = j {\partial }_{x} {m}^{ \tilde{u} }_{j-1} ( t,x ) {e}^{ \int_{t}^{T} {A}_{v} dv }
  = j (j-1) {m}^{ \tilde{u} }_{j-2} ( t,x ) {e}^{ 2 \int_{t}^{T} {A}_{v} dv }.
\end{aligned} \right.
\label{eq:partial derivatives of moment functions :eq}
\end{equation}
Since ${X}^{ t,x, \tilde{u} }_{T} - \mathbb{E}_{t} [ {X}^{ t,x, \tilde{u} }_{T} ] \sim N ( 0, {y}^{\beta }_{t} )$,
we can show that
\begin{equation}
\sum_{j=0}^{i} \binom{i}{j} (-1)^{i-j} \big( {m}^{ \tilde{u} }_{1} ( t,x ) \big)^{i-j} {m}^{ \tilde{u} }_{j} ( t,x )
\equiv {M}^{ \tilde{u} }_{i} ( t,x ) = {\alpha }_{i} ( {y}^{\beta }_{t} ), \quad \forall i \in \mathbb{N}.
\label{eq:central moment :eq}
\end{equation}
Consequently, $\phi \big( t, \vec{m}^{ \tilde{u} } ( t,x ) \big) = \kappa {m}^{ \tilde{u} }_{1} ( t,x ) + \psi \big( t, \vec{\alpha } ( {y}^{\beta }_{t} ) \big)$ follows from \cref{eq:objective functional particular :eq},
and hence,
\begin{equation*}
  \big\langle {\phi }_{ \vec{z} } \big( t, \vec{m}^{ \tilde{u} } ( t,x ) \big), \vec{m}^{ \tilde{u} }_{x} ( t,x ) \big\rangle
= {\partial }_{x} \phi \big( t, \vec{m}^{ \tilde{u} } ( t,x ) \big)
= \kappa {\partial }_{x} {m}^{ \tilde{u} }_{1} ( t,x )
= \kappa {e}^{ \int_{t}^{T} {A}_{v} dv },
\end{equation*}
while it follows from \cref{eq:objective functional rearrangement specific :eq}, \cref{eq:partial derivatives of moment functions :eq,eq:central moment :eq} that
\begin{align*}
    \big\langle {\phi }_{ \vec{z} } \big( t, \vec{m}^{ \tilde{u} } ( t,x ) \big), \vec{m}^{ \tilde{u} }_{xx} ( t,x ) \big\rangle
& = \sum_{j=2}^{n} \sum_{i=j}^{n} \binom{i}{j} \big( - {m}^{ \tilde{u} }_{1} ( t,x ) \big)^{i-j}
                                  {\psi }_{ {z}_{i} } \big( t, \vec{\alpha } ( {y}^{\beta }_{t} ) \big)
                                  {\partial }_{x}^{2} {m}^{ \tilde{u} }_{j} ( t,x ) \\
& = {e}^{ 2 \int_{t}^{T} {A}_{v} dv }
    \sum_{i=2}^{n} i (i-1) {\alpha }_{i-2} ( {y}^{\beta }_{t} ) {\psi }_{ {z}_{i} } \big( t, \vec{\alpha } ( {y}^{\beta }_{t} ) \big) \\
& = {e}^{ 2 \int_{t}^{T} {A}_{v} dv }
    \sum_{ 1 \le j \le \frac{n}{2} } 2j (2j-1) {\alpha }_{2j-2} ( {y}^{\beta }_{t} ) {\psi }_{ {z}_{2j} } \big( t, \vec{\alpha } ( {y}^{\beta }_{t} ) \big).
\end{align*}
Therefore, $\langle {\phi }_{ \vec{z} } ( t, \vec{m} ( t,x ) ), \vec{m}^{ \tilde{u} }_{xx} ( t,x ) \rangle < 0$ immediately arises from the condition \cref{eq:concavity condition for CNEC particular :eq}.
On the other hand, it follows from the equation \cref{eq:integral equation for CNEC particular :eq} that
\begin{equation*}
- \frac{ {B}_{t} }{ | {D}_{t} |^{2} }
  \frac{ \big\langle {\phi }_{ \vec{z} } \big( t, \vec{m} ( t,x ) \big), \vec{m}_{x} ( t,x ) \big\rangle }
       { \big\langle {\phi }_{ \vec{z} } \big( t, \vec{m} ( t,x ) \big), \vec{m}_{xx} ( t,x ) \big\rangle }
= {\beta }_{t} {e}^{ - \int_{t}^{T} {A}_{v} dv }.
\end{equation*}
Hence, $( \tilde{u}, \vec{m}^{ \tilde{u} } )$ determined by \cref{eq:CNEC particular :eq} solves the PDE system \cref{eq:EHJBS specific :eq}.
\end{proof}
 
\begin{remark}\label{rem:independence for odd central moments}
In view of \cref{eq:integral equation for CNEC particular :eq,eq:concavity condition for CNEC particular :eq}, one can find that 
corresponding to the CNEC $\tilde{u}$ given by \cref{eq:CNEC particular :eq}, ${X}^{ t,x, \tilde{u} }_{T}$ is normally distributed, and hence all of its $( 2j+1 )$-th conditional central moment ${\alpha }_{i} ( {y}_{t}^{\beta } )$ vanish.
Consequently, $\tilde{u}$ does not vary as the preferences on those odd-order central moments change.
In other words, if $\tilde{u}$ given by \cref{eq:integral equation for CNEC particular :eq,eq:concavity condition for CNEC particular :eq,eq:CNEC particular :eq}
is a CNEC for the objective function \cref{eq:objective functional particular :eq},
then it is also a CNEC for
\begin{equation*}
J ( t,x,u ) = \kappa \mathbb{E}_{t} [ {X}^{t,x,u}_{T} ] + \psi \big( t, {M}^{u}_{2} ( t,x ), 0, {M}^{u}_{4} ( t,x ), 0, \ldots, {1}_{ \{ n \in 2 \mathbb{N} \} } {M}^{u}_{n} ( t,x ) \big),
\end{equation*}
Furthermore, for any odd polynomial function $P ( t, \cdot )$, $\tilde{u}$ is also a CNEC for
\begin{equation*}
J ( t,x,u ) = \kappa \mathbb{E}_{t} [ {X}^{t,x,u}_{T} ] + \psi \big( t, \vec{M}^{u}_{-1} ( t,x ) \big) + \mathbb{E}_{t} \big[ P ( t, {X}^{t,x,u}_{T} - \mathbb{E}_{t} [ {X}^{t,x,u}_{T} ] ) \big].
\end{equation*}
Intuitively speaking, the independence of the preferences on the odd-order central moments for the CNEC $\tilde{u}$ meets the risk attitude 
that a positive odd-order central moment has both benefits (e.g., a long right tail) and drawbacks (e.g., a small mode or median), 
as we have mentioned in \Cref{sec:Introduction}.
As far as results are concerned, the CNEC does not only sacrifice those benefits, but also avoids those drawbacks, so that agents with different risk attitudes towards odd-order central moments can get the same time-consistent strategy.
However, this explanation strongly relies on the artificial structure of our objective function \cref{eq:objective functional particular :eq}, and is likely to fail for more general problems. 
For instance, if we only change the constant $\kappa $ therein to some function $\kappa (x)$, 
then the proof of \Cref{thm:CNEC particular} no longer holds, and one may obtain an $x$-dependent $\tilde{u} ( t,x )$ distinct from \cref{eq:CNEC particular :eq}.
Thus, this leads to a loss of the independence of the preferences on the odd-order central moments for the CNEC $\tilde{u}$.
Summing up, for the sake of rigor, we prefer to think of this independence as a mathematical consequence, rather than a necessary result for some intuitive characteristics.
\end{remark}

If $\kappa = 0$, then $\beta \equiv 0$ solves \cref{eq:integral equation for CNEC particular :eq}.
Moreover, if ${\psi }_{ {z}_{2} } ( t, \vec{0} ) < 0$ arising from \cref{eq:concavity condition for CNEC particular :eq} holds for every $t$, 
then $\tilde{u} ( t,x ) = - {F}_{t} / {D}_{t}$ gives a CNEC such that the corresponding ${X}_{T}$ is deterministic.
The case with $\kappa > 0$ is slightly more complicated.
The following theorem provides a sufficient condition for the existence and uniqueness of solution for \cref{eq:integral equation for CNEC particular :eq} with $\kappa > 0$.

\begin{theorem}\label{thm:uniqueness of solution to integral equation}
Suppose that
\begin{equation*}
f ( t,y ) := - \bigg( \sum_{ 1 \le j \le \frac{n}{2} } 2j ( 2j-1 ) {\alpha }_{2j-2} (y) {\psi }_{ {z}_{2j} } \big( t, \vec{\alpha } (y) \big) \bigg)^{-1}
\end{equation*}
is strictly positive and uniformly bounded over $[ 0,T ] \times [ 0, \infty )$.
Then, \cref{eq:integral equation for CNEC particular :eq} with $\kappa > 0$ admits a unique classical solution $\beta $.
\end{theorem}

\begin{proof}
\cref{eq:integral equation for CNEC particular :eq} can be re-expressed as $| {D}_{t} |^{2} {\beta }_{t} = \kappa {B}_{t} f ( t, {y}^{\beta }_{t} )$.
If $\beta $ solves it, then
\begin{align*}
& 0 \le {\beta }_{t} = \kappa {B}_{t} | {D}_{t} |^{-2} f ( t, {y}^{\beta }_{t} ) \le \kappa {B}_{t} | {D}_{t} |^{-2} \sup_{ y \ge 0 } f ( t,y ) < \infty, \\
& 0 \le {y}^{\beta }_{t} \le \int_{0}^{T} | {D}_{s} {\beta }_{s} |^{2} ds \le T {\kappa }^{2} ( \inf D )^{-2} ( \sup |B|^{2} ) ( \sup |f|^{2} ) < \infty,
\end{align*}
and $y = {y}^{\beta }$ solves ${y}_{t}' + | \kappa {B}_{t} |^{2} | {D}_{t} |^{-2} | f ( t, {y}_{t} ) |^{2} = 0$.
Conversely, we conclude that $| f ( t, \cdot ) |^{2}$ is Lipschitz continuous on $[ 0, T {\kappa }^{2} ( \inf D )^{-2} ( \sup |B|^{2} ) ( \sup |f|^{2} ) ]$, since
\begin{equation*}
\big| \big( f ( t,y ) \big)^{2} - \big( f ( t,z ) \big)^{2} \big| = | f ( t,y ) f ( t,z ) | | f ( t,y ) + f ( t,z ) | \Big| \frac{1}{ f ( t,y ) } - \frac{1}{ f ( t,z ) } \Big|,
\quad \forall y,z \ge 0.
\end{equation*}
Then, due to the existence and uniqueness theorem for first-order ordinary differential equation, there exists a unique $\{ {y}_{t} \}_{ t \in [ 0,T ] }$ such that
$| {D}_{t} |^{2} {y}_{t}' + | \kappa {B}_{t} f ( t, {y}_{t} ) |^{2} = 0$ with ${y}_{T} = 0$.
Hence, ${\beta }_{t} = \kappa {B}_{t} | {D}_{t} |^{-2} f ( t, {y}_{t} )$ gives the unique solution to (\ref{eq:integral equation for CNEC particular :eq}).
\end{proof}

Apart from the result given by \Cref{thm:CNEC particular},
according to \Cref{lem:Feynman-Kac representation} and \Cref{thm:Bellman equation}, seeking a CNEC can be reduced to solving the following Bellman equation:
\begin{align}
0 = \max_{ \zeta \in \mathbb{R} }
    \bigg\{ & {V}_{t} ( t,x )
            + {V}_{x} ( t,x ) ( {A}_{t} x + {B}_{t} u + {C}_{t} )
            + \frac{1}{2} {V}_{xx} ( t,x ) ( {D}_{t} \zeta + {F}_{t} )^{2} \notag \\
          & - {\phi }_{t} \big( t, \vec{m}^{ \tilde{u} } ( t,x ) \big)
            - \frac{1}{2} ( {D}_{t} \zeta + {F}_{t} )^{2}
              \big\langle \vec{m}^{ \tilde{u} }_{x} ( t,x ) {\phi }_{ \vec{z} \vec{z} } \big( t, \vec{m}^{ \tilde{u} } ( t,x ) \big), \vec{m}^{ \tilde{u} }_{x} ( t,x ) \big\rangle \bigg\},
\label{eq:Bellman equation for specific problem :eq}
\end{align}
subject to the terminal condition $V ( T,x ) = \kappa x + \psi ( T, \vec{0} )$.

\begin{theorem}\label{thm:solving Bellman equation for specific problem}
Suppose that the integral equation \cref{eq:integral equation for CNEC particular :eq} admits the solution $\beta $ such that \cref{eq:concavity condition for CNEC particular :eq} holds,
$\int_{t}^{T} {B}_{s} {\beta }_{s} ds < \infty $ and ${y}^{\beta }_{t} \equiv \int_{t}^{T} | {D}_{s} {\beta }_{s} |^{2} ds < \infty $.
Then,
\begin{equation}
V ( t,x ) = \kappa \Theta ( t,x ) + \kappa \int_{t}^{T} {B}_{s} {\beta }_{s} ds + \psi \big( t, \vec{\alpha } ( {y}^{\beta }_{t} ) \big)
\label{eq:value function a.w. specific problem :eq}
\end{equation}
fulfills \cref{eq:Bellman equation for specific problem :eq}, and \cref{eq:CNEC particular :eq} realizes the maximum on the right-hand side of \cref{eq:Bellman equation for specific problem :eq}.
\end{theorem}

\begin{proof}
For twice differentiable functions $H ( \vec{z} )$ and $h ( \vec{z} )$ with
\begin{equation*}
H ( \vec{z} ) = h \bigg( {z}_{1}, {z}_{2} - {z}_{1}^{2}, {z}_{3} - 3 {z}_{2} {z}_{1} + 2 {z}_{1}^{3}, \ldots, \sum_{j=0}^{n} \binom{n}{j} (-1)^{j} {z}_{1}^{j} {z}_{n-j} \bigg),
\end{equation*}
by straightforward calculation, one can obtain
\begin{equation*}
\sum_{i=1}^{n} \sum_{j=1}^{n} i j {H}_{ {z}_{i} {z}_{j} } ( \vec{z} ) {z}_{i-1} {z}_{j-1}
= {h}_{ {z}_{1} {z}_{1} } - \sum_{k=2}^{n} k (k-1) {h}_{ {z}_{k} } \sum_{i=0}^{k-2} \binom{k-2}{i} (-1)^{k-2-i} {z}_{1}^{k-2-i} {z}_{i},
\end{equation*}
where ${h}_{ {z}_{1} {z}_{1} } = {h}_{ {z}_{1} {z}_{1} } ( {z}_{1}, {z}_{2} - {z}_{1}^{2}, \ldots, \sum_{j=0}^{n} \binom{n}{j} (-1)^{j} {z}_{1}^{j} {z}_{n-j} )$ and so on for brevity.
Applying this differentiation result to \cref{eq:objective functional rearrangement specific :eq}, with the first identity in \cref{eq:partial derivatives of moment functions :eq}, we have
\begin{align*}
& \big\langle \vec{m}^{ \tilde{u} }_{x} ( t,x ) {\phi }_{ \vec{z} \vec{z} } \big( t, \vec{m}^{ \tilde{u} } ( t,x ) \big), \vec{m}^{ \tilde{u} }_{x} ( t,x ) \big\rangle \\
& = {e}^{ 2 \int_{t}^{T} {A}_{v} dv } \sum_{i=1}^{n} \sum_{j=1}^{n} i j {\phi }_{ {z}_{i} {z}_{j} } \big( t, \vec{m}^{ \tilde{u} } ( t,x ) \big) {m}^{ \tilde{u} }_{i-1} ( t,x ) {m}^{ \tilde{u} }_{j-1} ( t,x ) \\
& = - {e}^{ 2 \int_{t}^{T} {A}_{v} dv } \sum_{k=2}^{n} k (k-1) {M}^{ \tilde{u} }_{k-2} ( t,x ) {\psi }_{ {z}_{k} } \big( t, \vec{M}^{ \tilde{u} }_{-1} ( t,x ) \big) \\
& = - {e}^{ 2 \int_{t}^{T} {A}_{v} dv } \sum_{ 1 \le j \le \frac{n}{2} } 2j (2j-1) {\alpha }_{2j-2} ( {y}^{\beta }_{t} ) {\psi }_{ {z}_{2j} } \big( t, \vec{\alpha } ( {y}^{\beta }_{t} ) \big).
\end{align*}
Then, plugging the above result with ${\phi }_{t} ( t, \vec{m}^{ \tilde{u} } ( t,x ) ) = {\psi }_{t} ( t, \vec{M}^{ \tilde{u} }_{-1} ( t,x ) ) = {\psi }_{t} ( t, \vec{\alpha } ( {y}^{\beta }_{t} ) )$,
$\frac{1}{2} \kappa {B}_{t} {\beta }_{t} = {\partial }_{t} \psi \big( t, \vec{\alpha } ( {y}^{\beta }_{t} ) \big) - {\psi }_{t} \big( t, \vec{\alpha } ( {y}^{\beta }_{t} ) \big)$
arising from \cref{eq:integral equation for CNEC particular :eq}
and the EVF candidate \cref{eq:value function a.w. specific problem :eq} into the Bellman equation \cref{eq:Bellman equation for specific problem :eq}, yields
\begin{align}
0 = \max_{ \zeta \in \mathbb{R} }
    \bigg\{ & {e}^{ \int_{t}^{T} {A}_{v} dv } \Big( \zeta + \frac{ {F}_{t} }{ {D}_{t} } \Big) \kappa {B}_{t} - \frac{1}{2} \kappa {B}_{t} {\beta }_{t} \notag \\
          & + \frac{1}{2} {e}^{ 2 \int_{t}^{T} {A}_{v} dv } \Big( \zeta + \frac{ {F}_{t} }{ {D}_{t} } \Big)^{2} | {D}_{t} |^{2}
              \sum_{ 1 \le j \le \frac{n}{2} } 2j (2j-1) {\alpha }_{2j-2} ( {y}^{\beta }_{t} ) {\psi }_{ {z}_{2j} } \big( t, \vec{\alpha } ( {y}^{\beta }_{t} ) \big) \bigg\}.
\label{eq:expression of value function under substitution :eq}
\end{align}
Notably, the rest of this proof is to show the validness of \cref{eq:expression of value function under substitution :eq}
and that $\tilde{u} ( t,x )$ given by \cref{eq:CNEC particular :eq} realizes the maximum on the right-hand side of \cref{eq:expression of value function under substitution :eq}.
Since \cref{eq:concavity condition for CNEC particular :eq} holds, the maximization condition for \cref{eq:expression of value function under substitution :eq} can be re-expressed as
\begin{equation*}
0 = \frac{ \kappa {B}_{t} }{ | {D}_{t} |^{2} }
  + \Big( {\zeta }^{*} + \frac{ {F}_{t} }{ {D}_{t} } \Big) {e}^{ \int_{t}^{T} {A}_{v} dv }
    \sum_{ 1 \le j \le \frac{n}{2} } 2j (2j-1) {\alpha }_{2j-2} ( {y}^{\beta }_{t} ) {\psi }_{ {z}_{2j} } \big( t, \vec{\alpha } ( {y}^{\beta }_{t} ) \big),
\end{equation*}
where ${\zeta }^{*}$ represents the maximizer.
In comparison with \cref{eq:integral equation for CNEC particular :eq}, ${\zeta }^{*} = \tilde{u} ( t,x )$ follows.
Then, substituting the maximizer ${\zeta }^{*} = \tilde{u} ( t,x )$ back into \cref{eq:expression of value function under substitution :eq},
by using \cref{eq:integral equation for CNEC particular :eq} again, one can find that the right-hand side of \cref{eq:expression of value function under substitution :eq} vanishes.
So we complete this proof.
\end{proof}

\begin{example}\label{ex:linear combination of central moments}
For the following objective functions arising from mean-variance-skewness problems with arbitrarily fixed ${\kappa }_{2} > 0$ and ${\kappa }_{3}$
(which degenerates to a mean-variance objective function when ${\kappa }_{3} = 0$),
\begin{equation*}
{J}_{3} ( t,x,u ) = \mathbb{E}_{t} [ {X}^{t,x,u}_{T} ] + \sum_{j=2}^{3} (-1)^{j+1} \frac{ {\kappa }_{j} }{j!} \mathbb{E}_{t} \big[ ( {X}^{t,x,u}_{T} - \mathbb{E}_{t} [ {X}^{t,x,u}_{T} ] )^{j} \big],
\end{equation*}
\Cref{thm:CNEC particular} provides a $( {\kappa }_{3}, x )$-independent CNEC, as the following:
\begin{equation*}
\tilde{u}_{MV} ( t,x ) = \frac{ {B}_{t} }{ {\kappa }_{2} | {D}_{t} |^{2} } {e}^{ - \int_{t}^{T} {A}_{v} dv } - \frac{ {F}_{t} }{ {D}_{t} },
\end{equation*}
which due to \Cref{rem:independence for odd central moments} is also a CNEC for 
\begin{align*}
{J}_{MVS} ( t,x,u ) & = \mathbb{E}_{t} [ {X}^{t,x,u}_{T} ] - \frac{ {\kappa }_{2} }{2} \mathbb{E}_{t} \big[ ( {X}^{t,x,u}_{T} - \mathbb{E}_{t} [ {X}^{t,x,u}_{T} ] )^{2} \big] \\
                    & \quad + \frac{ {\kappa }_{3} }{6} \frac{ \mathbb{E}_{t} \big[ ( {X}^{t,x,u}_{T} - \mathbb{E}_{t} [ {X}^{t,x,u}_{T} ] )^{3} \big] }
                                                             { \big( \mathbb{E}_{t} \big[ ( {X}^{t,x,u}_{T} - \mathbb{E}_{t} [ {X}^{t,x,u}_{T} ] )^{2} \big] \big)^{ \frac{3}{2} } }.
\end{align*}
For the objective function related to mean-variance-skewness-kurtosis problems,
\begin{equation*}
{J}_{4} ( t,x,u ) = \mathbb{E}_{t} [ {X}^{t,x,u}_{T} ] + \sum_{j=2}^{4} (-1)^{j+1} \frac{ {\kappa }_{j} }{j!} \mathbb{E}_{t} \big[ ( {X}^{t,x,u}_{T} - \mathbb{E}_{t} [ {X}^{t,x,u}_{T} ] )^{j} \big],
\end{equation*}
where ${\kappa }_{2}, {\kappa }_{4} \ge 0$ with $| {\kappa }_{2} | + | {\kappa }_{4} | > 0$, the CNEC given by \Cref{thm:CNEC particular} is
\begin{equation*}
\tilde{u}_{4} ( t,x ) = \Big( {\kappa }_{2}^{3} + \frac{3}{2} {\kappa }_{4} {\theta }_{t} \Big)^{ - \frac{1}{3} } \frac{ {B}_{t} }{ | {D}_{t} |^{2} } {e}^{ - \int_{t}^{T} {A}_{v} dv } - \frac{ {F}_{t} }{ {D}_{t} },
\end{equation*}
which is also $( {\kappa }_{3}, x )$-independent.
In fact, according to \Cref{thm:CNEC particular} and the definition of $( {\theta }_{t}, {y}^{\beta }_{t} )$, 
we are supposed to solve ${\theta }_{t}' = {y}_{t}' ( {\kappa }_{2} + \frac{1}{2} {\kappa }_{4} {y}_{t} )^{2}$ for ${y}_{t}$.
Integrating the both sides over $[ t,T ]$ with rearranging the terms yields ${\kappa }_{2} + \frac{1}{2} {\kappa }_{4} {y}_{t} = ( {\kappa }_{2}^{3} + \frac{3}{2} {\kappa }_{4} {\theta }_{t} )^{1/3}$,
which leads to the desired result by differentiating the both sides w.r.t. $t$.
More generally, for the objective function
\begin{equation}
{J}_{n} ( t,x,u ) = \kappa \mathbb{E}_{t} [ {X}^{t,x,u}_{T} ] + \sum_{j=2}^{n} (-1)^{j+1} \frac{ {\kappa }_{j} }{j!} \mathbb{E}_{t} \big[ ( {X}^{t,x,u}_{T} - \mathbb{E}_{t} [ {X}^{t,x,u}_{T} ] )^{j} \big],
\label{eq:linear combination objective function :eq}
\end{equation}
as a linear combination of mean and central moments, where all ${\kappa }_{2j} \ge 0$ and at least one ${\kappa }_{2j} > 0$,
\Cref{thm:CNEC particular} suggests that seeking a CNEC can be reduced to solving
\begin{equation*}
0 = \Big| \frac{ \kappa {B}_{t} }{ {D}_{t} } \Big|^{2} + {y}_{t}' \bigg| \sum_{ 1 \le j \le \frac{n}{2} } \frac{ {\kappa }_{2j} }{ (2j-2)!! } ( {y}_{t} )^{j-1} \bigg|^{2}
\end{equation*}
as well as the integral equation ${y}_{t} = {y}^{\beta }_{t}$.
Integrating the both sides of the above differential equation over $[ t,T ]$ yields the algebraic equation
\begin{equation}
{\kappa }^{2} {\theta }_{t} = \int_{0}^{ {y}_{t} } \bigg| \sum_{ 0 \le j \le \frac{n}{2} - 1 } {\kappa }_{2j+2} \frac{ {z}^{j} }{ (2j)!!} \bigg|^{2} dz
\equiv 2 \int_{0}^{ \frac{1}{2} {y}_{t} } \bigg| \sum_{ 0 \le j \le \frac{n}{2} - 1 } {\kappa }_{2j+2} \frac{ {z}^{j} }{j!} \bigg|^{2} dz.
\label{eq:algebraic equation :eq}
\end{equation}
If $\kappa = 0$, then \cref{eq:algebraic equation :eq} has the unique solution $y \equiv 0$, which leads to $\beta \equiv 0$.
Otherwise, the unique solution ${y}_{t}$ of \cref{eq:algebraic equation :eq} is decreasing and continuously differentiable.
Then, \cref{eq:CNEC particular :eq} with ${\beta }_{t} = ( {D}_{t} )^{-1} \sqrt{ - {y}_{t}' }$ gives the CNEC.
\end{example}

\begin{example}
Interested readers may focus on the mean-variance-standardized moments objective function
\begin{align*}
{J}_{MVSM} ( t,x,u ) & = \mathbb{E}_{t} [ {X}^{t,x,u}_{T} ] - \frac{ {\kappa }_{2} }{2} \mathbb{E}_{t} \big[ ( {X}^{t,x,u}_{T} - \mathbb{E}_{t} [ {X}^{t,x,u}_{T} ] )^{2} \big] \\
                     & \quad + \sum_{j=3}^{n} (-1)^{j+1} \frac{ {\kappa }_{j} }{j!} \frac{ \mathbb{E}_{t} \big[ ( {X}^{t,x,u}_{T} - \mathbb{E}_{t} [ {X}^{t,x,u}_{T} ] )^{j} \big] }
                                                                                         { \big( \mathbb{E}_{t} \big[ ( {X}^{t,x,u}_{T} - \mathbb{E}_{t} [ {X}^{t,x,u}_{T} ] )^{2} \big] \big)^{ \frac{j}{2} } };
\end{align*}
that is, all higher-order central moments except the variance in \cref{eq:linear combination objective function :eq} are replaced by standardized moments of the corresponding order.
For this situation, we let $\psi ( t, {z}_{2}, \ldots, {z}_{n} ) = - \frac{ {\kappa }_{2} }{2} {z}_{2} + \sum_{j=3}^{n} (-1)^{j+1} \frac{ {\kappa }_{j} }{j!} {z}_{j} | {z}_{2} |^{ - j / 2 }$.
Consequently, the left-hand side of \cref{eq:integral equation for CNEC particular :eq} equals to ${B}_{t} | {D}_{t} |^{-2} - {\kappa }_{2} {\beta }_{t}$.
Therefore, $\tilde{u}_{MV}$ in the last \Cref{ex:linear combination of central moments} is also a CNEC for ${J}_{MVSM}$, if ${\kappa }_{2} > 0$.
\end{example}

\subsection{Open-loop Nash equilibrium control}

Now we show that the state-independent CNEC given by \cref{eq:CNEC particular :eq} is exactly a path-independent ONEC.
According to \Cref{thm:sufficient maximum principle for linear control problem}, seeking an ONEC can be reduced to solving the linear FBSDE
\begin{equation}
\left\{ \begin{aligned}
d \bar{X}_{s} & = ( {A}_{s} \bar{X}_{s} + {B}_{s} \bar{u}_{s} + {C}_{s} ) ds + ( {D}_{s} \bar{u}_{s} + {F}_{s} ) d {W}_{s}, ~ \forall s \in [ 0,T ], ~
  \bar{X}_{0} = {x}_{0}; \\
d {Y}^{t}_{s} & = - {A}_{s} {Y}^{t}_{s} ds + \mathcal{Y}^{t}_{s} d {W}_{s}, ~ \forall s \in [ t,T ], \\
  {Y}^{t}_{T} & = \kappa + \sum_{j=2}^{n} j {\psi }_{ {z}_{j} } \big( t, \vec{M}^{ \bar{u} }_{-1} ( t, \bar{X}_{t} ) \big)
                                          \big( ( \bar{X}_{T} - \mathbb{E}_{t} [ \bar{X}_{T} ] )^{j-1} - {M}^{ \bar{u} }_{j-1} ( t, \bar{X}_{t} ) \big); \\
d {Z}^{t}_{s} & = - 2 {A}_{s} {Z}^{t}_{s} ds + \mathcal{Z}^{t}_{s} d {W}_{s}, ~ \forall s \in [ t,T ], \\
  {Z}^{t}_{T} & = \sum_{j=2}^{n} j (j-1) {\psi }_{ {z}_{j} } \big( t, \vec{M}^{ \bar{u} }_{-1} ( t, \bar{X}_{t} ) \big) ( \bar{X}_{T} - \mathbb{E}_{t} [ \bar{X}_{T} ] )^{j-2},
\end{aligned} \right. \label{eq:FBSDE particular :eq}
\end{equation}
with the following optimality condition:
\begin{equation}
\lim_{s \downarrow t} \mathbb{E}_{t} [ {B}_{s} {Y}^{t}_{s} + {D}_{s} \mathcal{Y}^{t}_{s} ] = 0, \quad \lim_{s \downarrow t} \mathbb{E}_{t} [ {Z}^{t}_{s} ] \le 0,
\quad \mathbb{P}-a.s., ~ a.e. ~ t \in [ 0,T ).
\label{eq:SMP maximization condition particular :eq}
\end{equation}

\begin{lemma}\label{lem:solution to FBSDE particular}
Suppose that $\bar{u}_{t} = {\beta }_{t} \exp ( - \int_{t}^{T} {A}_{v} dv ) - {F}_{t} / {D}_{t}$, where $\beta $ is an arbitrarily fixed deterministic continuous function.
Then, \cref{eq:FBSDE particular :eq} admits the unique solution
\begin{equation*}
\left\{ \begin{aligned}
        \bar{X}_{s} & = {x}_{0} {e}^{ \int_{0}^{s} {A}_{v} dv } + \int_{0}^{s} {e}^{ \int_{\tau }^{s} {A}_{v} dv } \Big( {C}_{\tau } - \frac{ {B}_{\tau } }{ {D}_{\tau } } {F}_{\tau } \Big) d \tau \\
                    & \quad + {e}^{ - \int_{s}^{T} {A}_{v} dv } \bigg( \int_{0}^{s} {B}_{\tau } {\beta }_{\tau } d \tau + \int_{0}^{s} {D}_{\tau } {\beta }_{\tau } d {W}_{\tau } \bigg), \\
        {Y}^{t}_{s} & = {e}^{ \int_{s}^{T} {A}_{v} dv }
                        \Bigg( \kappa + \sum_{j=2}^{n} j {\psi }_{ {z}_{j} } \big( t, \vec{\alpha } ( {y}^{\beta }_{t} ) \big)
                                                       \bigg( \mathbb{E}_{s} \bigg[ \Big( \int_{t}^{T} {D}_{\tau } {\beta }_{\tau } d {W}_{\tau } \Big)^{j-1} \bigg] - {\alpha }_{j-1} ( {y}^{\beta }_{t} ) \bigg) \Bigg), \\
\mathcal{Y}^{t}_{s} & = {e}^{ \int_{s}^{T} {A}_{v} dv } {D}_{s} {\beta }_{s}
                        \sum_{j=2}^{n} j (j-1) {\psi }_{ {z}_{j} } \big( t, \vec{\alpha } ( {y}^{\beta }_{t} ) \big)
                                       \mathbb{E}_{s} \bigg[ \Big( \int_{t}^{T} {D}_{\tau } {\beta }_{\tau } d {W}_{\tau } \Big)^{j-2} \bigg], \\
        {Z}^{t}_{s} & = {e}^{ 2 \int_{s}^{T} {A}_{v} dv }
                        \sum_{j=2}^{n} j (j-1) {\psi }_{ {z}_{j} } \big( t, \vec{\alpha } ( {y}^{\beta }_{t} ) \big)
                                       \mathbb{E}_{s} \bigg[ \Big( \int_{t}^{T} {D}_{\tau } {\beta }_{\tau } d {W}_{\tau } \Big)^{j-2} \bigg], \\
\mathcal{Z}^{t}_{s} & = {e}^{ 2 \int_{s}^{T} {A}_{v} dv } {D}_{s} {\beta }_{s}
                        \sum_{j=3}^{n} j (j-1) (j-2) {\psi }_{ {z}_{j} } \big( t, \vec{\alpha } ( {y}^{\beta }_{t} ) \big)
                                       \mathbb{E}_{s} \bigg[ \Big( \int_{t}^{T} {D}_{\tau } {\beta }_{\tau } d {W}_{\tau } \Big)^{j-3} \bigg].
\end{aligned} \right.
\end{equation*}
\end{lemma}

\begin{proof}
The explicit expression of $\bar{X}$ arises from the SDE
\begin{equation*}
d \bar{X}_{s} = \Big( {A}_{s} \bar{X}_{s} + {B}_{s} {\beta }_{s} {e}^{ - \int_{s}^{T} {A}_{v} dv } + {C}_{s} - \frac{ {B}_{s} }{ {D}_{s} } {F}_{s} \Big) ds
              + {D}_{s} {\beta }_{s} {e}^{ - \int_{s}^{T} {A}_{v} dv } d {W}_{s}.
\end{equation*}
Then, we have $\bar{X}_{T} = \mathbb{E}_{t} [ \bar{X}_{T} ] + \int_{t}^{T} {D}_{s} {\beta }_{s} d {W}_{s}$,
and hence $\vec{M}^{ \bar{u} }_{-1} ( t, \bar{X}_{t} ) = \vec{\alpha } ( {y}^{\beta }_{t} )$. See also the proof for \Cref{thm:CNEC particular}.
Consequently, by the Lipschitz continuity of $\psi ( t, \cdot )$, we conclude that ${Y}^{t}_{T}$ and ${Z}^{t}_{T}$ are square-integrable.
Due to \cite[Theorem 7.2.2, p. 349]{Yong-Zhou-1999}, each backward stochastic differential equation (BSDE) in \cref{eq:FBSDE particular :eq} admits a unique square-integrable $\mathbb{F}$-predictable solution.
Furthermore, from the BSDEs in \cref{eq:FBSDE particular :eq} one can obtain
\begin{equation*}
{Y}^{t}_{s} = {e}^{ \int_{s}^{T} {A}_{v} dv }  \mathbb{E}_{s} [ {Y}^{t}_{T} ], \qquad
{Z}^{t}_{s} = {e}^{ 2 \int_{s}^{T} {A}_{v} dv } \mathbb{E}_{s} [ {Z}^{t}_{T} ],
\end{equation*}
which lead to the desired explicit expression of $( {Y}^{t}, {Z}^{t} )$.
In order to derive the explicit expression of $\mathcal{Y}^{t}$ from ${Y}^{t}_{s} \exp ( - \int_{s}^{T} {A}_{v} dv ) = \mathbb{E}_{s} [ {Y}^{t}_{T} ]$, we differentiate the both sides of
\begin{equation*}
  \sum_{j=0}^{\infty } \frac{ {z}^{j} }{j!} \mathbb{E}_{s} \big[ ( \bar{X}_{T} - \mathbb{E}_{t} [ \bar{X}_{T} ] )^{j} \big]
= \mathbb{E}_{s} \big[ {e}^{ z \int_{t}^{T} {D}_{v} {\beta }_{v} d {W}_{v} } \big]
= {e}^{ z \int_{t}^{s} {D}_{v} {\beta }_{v} d {W}_{v} + \frac{1}{2} {z}^{2} {y}^{\beta }_{s} }.
\end{equation*}
with respect to $s$ to show that
$d \mathbb{E}_{s} [ ( \bar{X}_{T} - \mathbb{E}_{t} [ \bar{X}_{T} ] )^{j} ] = j \mathbb{E}_{s} [ ( \bar{X}_{T} - \mathbb{E}_{t} [ \bar{X}_{T} ] )^{j-1} ] {D}_{s} {\beta }_{s} d {W}_{s}$.
Consequently,
\begin{equation*}
  d \Big( {Y}^{t}_{s} {e}^{ - \int_{s}^{T} {A}_{v} dv } \Big)
= \sum_{j=2}^{n} j (j-1) {\psi }_{ {z}_{j} } \big( t, \vec{\alpha } ( {y}^{\beta }_{t} ) \big) \mathbb{E}_{s} \big[ ( \bar{X}_{T} - \mathbb{E}_{t} [ \bar{X}_{T} ] )^{j-2} \big] {D}_{s} {\beta }_{s} d {W}_{s}.
\end{equation*}
By comparing $d ( {Y}^{t}_{s} \exp ( - \int_{s}^{T} {A}_{v} dv ) ) = \exp ( - \int_{s}^{T} {A}_{v} dv ) \mathcal{Y}^{t}_{s} d {W}_{s}$ from \cref{eq:FBSDE particular :eq} with the above-mentioned result,
we obtain the desired expression of $\mathcal{Y}^{t}$ according to the existence and uniqueness of $\mathcal{Y}^{t}$.
In the same manner, comparing
\begin{equation*}
  d \Big( {Z}^{t}_{s} {e}^{ - 2 \int_{s}^{T} {A}_{v} dv } \Big)
= \sum_{j=3}^{n} j (j-1) (j-2) {\psi }_{ {z}_{j} } \big( t, \vec{\alpha } ( {y}^{\beta }_{t} ) \big) \mathbb{E}_{s} \big[ ( \bar{X}_{T} - \mathbb{E}_{t} [ \bar{X}_{T} ] )^{j-3} \big] {D}_{s} {\beta }_{s} d {W}_{s}
\end{equation*}
with $d ( {Z}^{t}_{s} \exp ( - 2 \int_{s}^{T} {A}_{v} dv ) ) = \exp ( - 2 \int_{s}^{T} {A}_{v} dv ) \mathcal{Z}^{t}_{s} d {W}_{s}$ from \cref{eq:FBSDE particular :eq} yields the desired expression of $\mathcal{Z}^{t}$.
\end{proof}

\begin{theorem}\label{thm:ONEC particular}
Let $\beta $ fulfill \cref{eq:integral equation for CNEC particular :eq,eq:concavity condition for CNEC particular :eq}.
Then, for the objective function \cref{eq:objective functional particular :eq} and the controlled SDE \cref{eq:linear controlled SDE :eq},
$\{ {\beta }_{t} \exp ( - \int_{t}^{T} {A}_{v} dv ) - {F}_{t} / {D}_{t} \}_{ t \in [ 0,T ] }$ is an ONEC.
\end{theorem}

\begin{proof}
From \Cref{lem:solution to FBSDE particular}, we obtain
$\lim_{s \downarrow t} \mathbb{E}_{t} [ {B}_{s} {Y}^{t}_{s} + {D}_{s} \mathcal{Y}^{t}_{s} ] = {B}_{t} {Y}^{t}_{t} + {D}_{t} \mathcal{Y}^{t}_{t}$ and
$\lim_{s \downarrow t} \mathbb{E}_{t} [ {Z}^{t}_{s} ] = {Z}^{t}_{t}$,
where
\begin{equation*}
\left\{ \begin{aligned}
        {Y}^{t}_{t} & = \kappa {e}^{ \int_{t}^{T} {A}_{v} dv }, \\
\mathcal{Y}^{t}_{t} & = {e}^{ \int_{t}^{T} {A}_{v} dv } {D}_{t} {\beta }_{t}
                        \sum_{ 1 \le j \le \frac{n}{2} } 2j (2j-1) {\alpha }_{2j-2} ( {y}^{\beta }_{t} ) {\psi }_{ {z}_{2j} } \big( t, \vec{\alpha } ( {y}^{\beta }_{t} ) \big), \\
        {Z}^{t}_{t} & = {e}^{ 2 \int_{t}^{T} {A}_{v} dv }
                        \sum_{ 1 \le j \le \frac{n}{2} } 2j (2j-1) {\alpha }_{2j-2} ( {y}^{\beta }_{t} ) {\psi }_{ {z}_{2j} } \big( t, \vec{\alpha } ( {y}^{\beta }_{t} ) \big).
\end{aligned} \right.
\end{equation*}
${D}_{t} \mathcal{Y}^{t}_{t} = - {B}_{t} {Y}^{t}_{t}$ follows from \cref{eq:integral equation for CNEC particular :eq},
while ${Z}^{t}_{t} < 0$ follows from \cref{eq:concavity condition for CNEC particular :eq}.
Thus, for $\beta $ given by \cref{eq:integral equation for CNEC particular :eq,eq:concavity condition for CNEC particular :eq}, the optimality condition \cref{eq:SMP maximization condition particular :eq} holds.
According to \Cref{thm:sufficient maximum principle for linear control problem}, $\bar{u}$ is an ONEC.
\end{proof}

\subsection{Limiting cases}
\label{subsec: Limiting cases}

In the previous subsections, we have shown that
the univariate time function given by \cref{eq:integral equation for CNEC particular :eq,eq:concavity condition for CNEC particular :eq,eq:CNEC particular :eq}
is a CNEC and an ONEC for the control problem with the objective function \cref{eq:objective functional particular :eq} and the controlled SDE \cref{eq:linear controlled SDE :eq}.
In particular, for the objective function \cref{eq:linear combination objective function :eq} as a linear combination of mean and central moments that researchers may be more concerned about,
deriving the state-independent CNEC or the path-independent ONEC is reduced to solving the algebraic equation \cref{eq:algebraic equation :eq}.
However, the celebrated Abel-Ruffini Theorem and Galois Theory indicate that there is no solution in radicals to \cref{eq:algebraic equation :eq} of finite degree $n \ge 6$ with arbitrary coefficients ${\kappa }_{2j}$.
In this subsection, we let $n$ tend to infinity.
As a result, we provide a heuristic approach to find state-independent CNEC, as well as path-independent ONEC, for several artificial objective functions.
Notably, the theories in \Cref{sec:Sufficient condition of CNEC: extended HJB-PDE,sec:Sufficient condition of ONEC: maximum principle} are only suitable for the problem with finitely many higher-order moments,
so we also provide verification procedures by applying those theories to $\Psi ( t,x,y, \vec{z} ) \equiv \Phi ( y, {z}_{1} )$ (i.e. without any higher-order moment) for the sake of rigor.

\subsubsection{Limiting case I: mean-exponential utility function}
\label{subsubsec: Limiting case I: mean-exponential objective function}

Let us consider the objective function
\begin{equation}
J ( t,x,u ) = \kappa \mathbb{E}_{t} [ {X}^{t,x,u}_{T} ] + \sum_{j=2}^{\infty } \frac{ (-c)^{j-1} }{j!} \mathbb{E}_{t} \big[ ( {X}^{t,x,u}_{T} - \mathbb{E}_{t} [ {X}^{t,x,u}_{T} ] )^{j} \big]
\label{eq:objective function case I :eq}
\end{equation}
for fixed $\kappa \ge 0$ and $c > 0$, or equivalently,
\begin{equation}
J ( t,x,u ) = \kappa \mathbb{E}_{t} [ {X}^{t,x,u}_{T} ] - \frac{1}{c} \mathbb{E}_{t} \big[ {e}^{ - c ( {X}^{t,x,u}_{T} - \mathbb{E}_{t} [ {X}^{t,x,u}_{T} ] ) } \big] + \frac{1}{c}.
\label{eq:objective function case I equivalent :eq}
\end{equation}
Notably, \cref{eq:objective function case I equivalent :eq} can be regarded as a weighted average of
the expectation $\mathbb{E}_{t} [ {X}^{t,x,u}_{T} ]$ and the expected exponential utility for ${X}^{t,x,u}_{T} - \mathbb{E}_{t} [ {X}^{t,x,u}_{T} ]$, for which $c$ represents the CARA coefficient.
By \Cref{thm:CNEC particular,thm:ONEC particular}, to determine a CNEC and an ONEC for \cref{eq:objective function case I :eq},
it suffices to solve ${\kappa }^{2} {\theta }_{t} = \exp ( {c}^{2} {y}_{t} ) - 1$ arising from \cref{eq:linear combination objective function :eq,eq:algebraic equation :eq}.
As a result, ${y}_{t} = {c}^{-2} \ln ( 1 + {\kappa }^{2} {\theta }_{t} )$, and hence, ${\beta }_{t} = \kappa {c}^{-1} {B}_{t} | {D}_{t} |^{-2} ( 1 + {\kappa }^{2} {\theta }_{t} )^{-1/2}$ solves ${y}^{\beta }_{t} = {y}_{t}$.
Moreover, \cref{eq:concavity condition for CNEC particular :eq} for $n = \infty $ holds, since
\begin{equation*}
 \sum_{j=1}^{\infty } j (2j-1) {\alpha }_{2j-2} ( {y}^{\beta }_{t} ) {\psi }_{ {z}_{2j} } \big( t, \vec{\alpha } ( {y}^{\beta }_{t} ) \big)
= - \frac{c}{2} \sum_{j=1}^{\infty } \frac{ ( {c}^{2} {y}^{\beta }_{t} )^{j-1} }{(2j-2)!!}
= - \frac{c}{2} {e}^{ \frac{1}{2} {c}^{2} {y}^{\beta }_{t} } < 0.
\end{equation*}
Hence, for the CNEC $\tilde{u}$, the ONEC $\bar{u}$ and the EVF, we have
\begin{align}
\tilde{u} ( t,x ) & = \bar{u}_{t} = \frac{ \kappa {B}_{t} {e}^{ - \int_{t}^{T} {A}_{v} dv } }{ c | {D}_{t} |^{2} \sqrt{ 1 + {\kappa }^{2} {\theta }_{t} } } - \frac{ {F}_{t} }{ {D}_{t} },
\label{eq:NEC in case I :eq}
\\
V ( t,x ) & = J ( t, x, \tilde{u} ) = \kappa \Theta ( t,x ) + \frac{ \sqrt{ 1 + {\kappa }^{2} {\theta }_{t} } - 1 }{c}.
\label{eq:EVF in case I :eq}
\end{align}

The rest of this subsection is to show that $\tilde{u}$ (resp. $\bar{u}$) is exactly a CNEC (resp. an ONEC) for \cref{eq:objective function case I equivalent :eq}.
We may refer to the total of the following steps as a ``verification procedure''.
On the one hand, we let $\Phi ( t,x,y, \vec{z} ) \equiv \Phi ( y, {z}_{1} ) = \kappa y - {c}^{-1} {e}^{ - c ( y - {z}_{1} ) } + {c}^{-1}$ for \cref{eq:objective function case I equivalent :eq}.
According to \cref{eq:Bellman equation reduced :eq}, the Bellman equation for deriving a CNEC is
\begin{equation}
\begin{aligned}
0 = \max_{ \zeta \in \mathbb{R} }
    \bigg\{ & {V}_{t} ( t,x ) + {V}_{x} ( t,x ) ( {A}_{t} x + {B}_{t} \zeta + {C}_{t} ) + \frac{1}{2} {V}_{xx} ( t,x ) ( {D}_{t} \zeta + {F}_{t} )^{2} \\
          & - \Big( {\lambda }^{t,x}_{x} ( t,x ) {\partial }_{x} {m}^{ \tilde{u} }_{1} ( t,x )
                  + \frac{1}{2} {\mu }^{t,x} ( t,x ) \big( {\partial }_{x} {m}^{ \tilde{u} }_{1} ( t,x ) \big)^{2} \Big) ( {D}_{t} \zeta + {F}_{t} )^{2} \bigg\},
\end{aligned}
\label{eq:Bellman equation in case I :eq}
\end{equation}
where
\begin{equation}
\left\{ \begin{aligned}
{\lambda }^{s,y} ( t,x ) & = \mathbb{E}_{t} \big[ {\Phi }_{ {z}_{1} } \big( {X}^{ t,x, \tilde{u} }_{T}, {m}^{ \tilde{u} }_{1} ( s,y ) \big) \big]
                           = - \mathbb{E}_{t} \big[ {e}^{ - c ( {X}^{ t,x, \tilde{u} }_{T} - {m}^{ \tilde{u} }_{1} ( s,y ) ) } \big], \\
    {\mu }^{s,y} ( t,x ) & = \mathbb{E}_{t} \big[ {\Phi }_{ {z}_{1} {z}_{1} } \big( {X}^{ t,x, \tilde{u} }_{T}, {m}^{ \tilde{u} }_{1} ( s,y ) \big) \big]
                           = c {\lambda }^{s,y} ( t,x ).
\end{aligned} \right.
\label{eq:lambda-mu in case I :eq}
\end{equation}
Now we show that $V$ given by \cref{eq:EVF in case I :eq} fulfills \cref{eq:Bellman equation in case I :eq}, whereby $\tilde{u}$ given by \cref{eq:NEC in case I :eq} realizes the maximum.
Corresponding to $\tilde{u}$ given by \cref{eq:NEC in case I :eq}, we have 
\begin{equation*}
  {\lambda }^{s,y}_{x} ( t,x )
= c {e}^{ \int_{t}^{T} {A}_{v} dv } \mathbb{E}_{t} \big[ {e}^{ - c ( {X}^{ t,x, \tilde{u} }_{T} - {m}^{ \tilde{u} }_{1} ( s,y ) ) } \big]
= - c {e}^{ \int_{t}^{T} {A}_{v} dv } {\lambda }^{s,y} ( t,x )
\end{equation*}
and ${\lambda }^{t,x} ( t,x ) = - \mathbb{E}_{t} [ \exp ( - c \int_{t}^{T} {D}_{s} {\beta }_{s} d {W}_{s} ) ] = - \sqrt{ 1 + {\kappa }^{2} {\theta }_{t} }$.
It follows from \cref{eq:EVF in case I :eq} with the above-mentioned results that the right-hand side of \cref{eq:Bellman equation in case I :eq} is equal to
\begin{align*}
\max_{ \zeta \in \mathbb{R} }
\bigg\{ - \frac{ {\kappa }^{2} | {B}_{t} |^{2} }{ 2 c | {D}_{t} |^{2} \sqrt{ 1 + {\kappa }^{2} {\theta }_{t} } }
      + \kappa {e}^{ \int_{t}^{T} {A}_{v} dv } \frac{ {B}_{t} }{ {D}_{t} } ( {D}_{t} \zeta + {F}_{t} ) & \\
      - \frac{c}{2} \sqrt{ 1 + {\kappa }^{2} {\theta }_{t} } {e}^{ 2 \int_{t}^{T} {A}_{v} dv } ( {D}_{t} \zeta + {F}_{t} )^{2} & \bigg\} = 0,
\end{align*}
for which ${\zeta }^{*} = \tilde{u} ( t,x )$ given by \cref{eq:NEC in case I :eq} satisfies the first-order derivative condition:
\begin{equation*}
0 = \kappa {e}^{ \int_{t}^{T} {A}_{v} dv } {B}_{t} - c \sqrt{ 1 + {\kappa }^{2} {\theta }_{t} } {e}^{ 2 \int_{t}^{T} {A}_{v} dv } {D}_{t} ( {D}_{t} {\zeta }^{*} + {F}_{t} ).
\end{equation*}
So we are done.
On the other hand, to show that $\bar{u}$ given by \cref{eq:NEC in case I :eq} is an ONEC for \cref{eq:objective function case I equivalent :eq},
it suffices to show the validity of \cref{eq:SMP maximization condition particular :eq} with $( {Y}^{t}, \mathcal{Y}^{t}, {Z}^{t} )$ given by the BSDEs
\begin{equation*}
\left\{ \begin{aligned}
d {Y}^{t}_{s} & = - {A}_{s} {Y}^{t}_{s} ds + \mathcal{Y}^{t}_{s} d {W}_{s}, \quad \forall s \in [ t,T ], \\
  {Y}^{t}_{T} & = {\Psi }_{y} ( \bar{X}_{T}, \mathbb{E}_{t} [ \bar{X}_{T} ] )
                + \mathbb{E}_{t} \big[ {\Psi }_{ {z}_{1} } ( \bar{X}_{T}, \mathbb{E}_{t} [ \bar{X}_{T} ] ) \big]
                = \kappa + {e}^{ - c \int_{t}^{T} {D}_{v} {\beta }_{v} d {W}_{v} } - {e}^{ \frac{1}{2} {c}^{2} {y}^{\beta }_{t} }; \\
d {Z}^{t}_{s} & = - 2 {A}_{s} {Z}^{t}_{s} ds + \mathcal{Z}^{t}_{s} d {W}_{s}, \quad \forall s \in [ t,T ], \\
  {Z}^{t}_{T} & = {\Psi }_{yy} ( \bar{X}_{T}, \mathbb{E}_{t} [ \bar{X}_{T} ] )
                = - c {e}^{ - c \int_{t}^{T} {D}_{v} {\beta }_{v} d {W}_{v} },
\end{aligned} \right.
\end{equation*}
according to \Cref{thm:sufficient maximum principle for linear control problem} with $\bar{X}_{T} - \mathbb{E}_{t} [ \bar{X}_{T} ] = \int_{t}^{T} {D}_{s} {\beta }_{s} d {W}_{s}$.
From
\begin{equation*}
\left\{ \begin{aligned}
& {Y}^{t}_{s} {e}^{ - \int_{s}^{T} {A}_{v} dv }
  = \kappa + \mathbb{E}_{s} \big[ {e}^{ - c \int_{t}^{T} {D}_{v} {\beta }_{v} d {W}_{v} } \big] - {e}^{ \frac{1}{2} {c}^{2} {y}^{\beta }_{t} }, \\
& d \big( {Y}^{t}_{s} {e}^{ - \int_{s}^{T} {A}_{v} dv } \big) = {e}^{ - \int_{s}^{T} {A}_{v} dv } \mathcal{Y}^{t}_{s} d {W}_{s},
\end{aligned} \right.
\end{equation*}
we obtain $\mathcal{Y}^{t}_{s} = - c \exp ( \int_{s}^{T} {A}_{v} dv + \frac{1}{2} {c}^{2} {y}^{\beta }_{t} - c \int_{t}^{s} {D}_{v} {\beta }_{v} d {W}_{v} ) {D}_{s} {\beta }_{s}$.
In the same manner, we can show that ${Z}^{t}_{s} = - c \exp ( 2 \int_{s}^{T} {A}_{v} dv ) \mathbb{E}_{s} [ \exp ( - c \int_{t}^{T} {D}_{v} {\beta }_{v} d {W}_{v} ) ]$.
Consequently,
\begin{equation*}
  \lim_{s \downarrow t} \mathbb{E}_{t} [ {B}_{s} {Y}^{t}_{s} + {D}_{s} \mathcal{Y}^{t}_{s} ]
= {B}_{t} {Y}^{t}_{t} + {D}_{t} \mathcal{Y}^{t}_{t} = 0, \quad
  \lim_{ s \downarrow t } \mathbb{E}_{t} [ {Z}^{t}_{s} ] = {Z}^{t}_{t} < 0.
\end{equation*}
Therefore, \cref{eq:SMP maximization condition particular :eq} holds, and hence $\bar{u}$ given by \cref{eq:NEC in case I :eq} is an ONEC for \cref{eq:objective function case I equivalent :eq}.

\subsubsection{Limiting case II: hyperbolic cosine penalty function}
\label{subsubsec: Limiting case II: hyperbolic cosine penalty function}

For arbitrarily fixed $\kappa \ge 0$ and $c > 0$, we consider the objective function
\begin{equation}
J ( t,x,u ) = \kappa \mathbb{E}_{t} [ {X}^{t,x,u}_{T} ] - \frac{1}{c} \mathbb{E}_{t} \big[ \cosh \big( c ( {X}^{t,x,u}_{T} - \mathbb{E}_{t} [ {X}^{t,x,u}_{T} ] ) \big) \big] + \frac{1}{c},
\label{eq:objective function case II equivalent :eq}
\end{equation}
where the hyperbolic cosine function is the penalty function for the deviation ${X}^{t,x,u}_{T} - \mathbb{E}_{t} [ {X}^{t,x,u}_{T} ]$.
Roughly speaking, \cref{eq:objective function case II equivalent :eq} is an analogue to the conventional mean-variance objective function,
but it replaces the quadratic function for variance with a hyperbolic cosine function.
Notably, \cref{eq:objective function case II equivalent :eq} can be re-expressed as
\begin{equation}
J ( t,x,u ) = \kappa \mathbb{E}_{t} [ {X}^{t,x,u}_{T} ] - \sum_{j=1}^{\infty } \frac{ {c}^{2j-1} }{ (2j)!} \mathbb{E}_{t} \big[ ( {X}^{t,x,u}_{T} - \mathbb{E}_{t} [ {X}^{t,x,u}_{T} ] )^{2j} \big],
\label{eq:objective function case II :eq}
\end{equation}
which is the result of removing all odd-order central moments from \cref{eq:objective function case I :eq}.
This implies that $\tilde{u}$ (resp. $\bar{u}$) given by \cref{eq:NEC in case I :eq} is a CNEC (resp. an ONEC) for \cref{eq:objective function case II :eq},
and $V$ given by \cref{eq:EVF in case I :eq} is the EVF.
Interested readers can follow our verification procedure as in the previous \Cref{subsubsec: Limiting case I: mean-exponential objective function}
to see that $\tilde{u}$ (resp. $\bar{u}$) is exactly a CNEC (resp. an ONEC) for \cref{eq:objective function case II equivalent :eq}.
Here we omit the statement due to the page limit.

\subsubsection{Limiting case III: cosine penalty function}
\label{subsubsec: Limiting case III: cosine penalty function}

In this subsection, we employ an irrational penalty function $S(x) = ( 1 - \cos (cx) ) / c$ with $c > 0$; that is,
\begin{equation}
J ( t,x,u ) = \kappa \mathbb{E}_{t} [ {X}^{t,x,u}_{T} ] + \frac{1}{c} \mathbb{E}_{t} \big[ \cos \big( c ( {X}^{t,x,u}_{T} - \mathbb{E}_{t} [ {X}^{t,x,u}_{T} ] ) \big) \big] - \frac{1}{c}.
\label{eq:objective function case III equivalent :eq}
\end{equation}
Some results in the following are parallel to those in \Cref{subsubsec: Limiting case II: hyperbolic cosine penalty function}, because of the link between $\cos ( \cdot )$ and $\cosh ( \cdot )$.
In addition, for technic reasons we assume that ${\kappa }^{2} {\theta }_{0} < 1$,
so that the marginal reward from $\mathbb{E}_{t} [ {X}^{t,x,u}_{T} ]$ and the penalty for the deviation ${X}^{t,x,u}_{T} - \mathbb{E}_{t} [ {X}^{t,x,u}_{T} ]$ are both bounded.
Let us rewrite \cref{eq:objective function case III equivalent :eq} as
\begin{equation}
J ( t,x,u ) = \kappa \mathbb{E}_{t} [ {X}^{t,x,u}_{T} ] - \sum_{j=1}^{\infty } \frac{ (-1)^{j-1} {c}^{2j-1} }{ (2j)! } \mathbb{E}_{t} \big[ ( {X}^{t,x,u}_{T} - \mathbb{E}_{t} [ {X}^{t,x,u}_{T} ] )^{2j} \big].
\label{eq:objective function case III :eq}
\end{equation}
Following the same line as deriving \cref{eq:NEC in case I :eq,eq:EVF in case I :eq},
we obtain ${y}^{\beta }_{t} = - {c}^{-2} \ln ( 1 - {\kappa }^{2} {\theta }_{t} )$, which solves ${\kappa }^{2} {\theta }_{t} = 1 - \exp ( - {c}^{2} {y}_{t} )$.
Then, ${\beta }_{t} = \kappa {c}^{-1} {B}_{t} | {D}_{t} |^{-2} ( 1 - {\kappa }^{2} {\theta }_{t} )^{-1/2}$ and
\begin{equation*}
 \sum_{j=1}^{\infty } j (2j-1) {\alpha }_{2j-2} ( {y}^{\beta }_{t} ) {\psi }_{ {z}_{2j} } \big( t, \vec{\alpha } ( {y}^{\beta }_{t} ) \big)
= - \frac{c}{2} \sum_{j=1}^{\infty } \frac{ ( - {c}^{2} {y}^{\beta }_{t} )^{j-1} }{(2j-2)!!}
= - \frac{c}{2} {e}^{ - \frac{1}{2} {c}^{2} {y}^{\beta }_{t} } < 0.
\end{equation*}
Therefore, for the CNEC $\tilde{u}$, the ONEC $\bar{u}$ and the EVF $V$, we have
\begin{align*}
\tilde{u} ( t,x ) = \bar{u}_{t} = \frac{ \kappa {B}_{t} {e}^{ - \int_{t}^{T} {A}_{v} dv } }{ c | {D}_{t} |^{2} \sqrt{ 1 - {\kappa }^{2} {\theta }_{t} } } - \frac{ {F}_{t} }{ {D}_{t} }, \quad
V ( t,x ) = \kappa \Theta ( t,x ) + \frac{ 1 - \sqrt{ 1 - {\kappa }^{2} {\theta }_{t} } }{c}.
\end{align*}
Here we omit the statement for the verification procedure due to the page limit, and encourage interested readers to try it on their own.
In fact, the verification procedure for a more general case will be provided in the coming \Cref{subsubsec: Limiting case IV: ambiguous cosine penalty function}.

\subsubsection{Limiting case IV: ambiguous cosine penalty function}
\label{subsubsec: Limiting case IV: ambiguous cosine penalty function}

Let $H$ be an $\mathbb{F}$-independent random variable whose all original moments of positive integer order exist, and
$S(x) = 1 - \mathbb{E} [ \cos ( H x ) ]$ be the penalty function for
\begin{equation}
\begin{aligned}
J ( t,x,u ) & = \kappa \mathbb{E}_{t} [ {X}^{t,x,u}_{T} ] + \mathbb{E}_{t} \big[ \cos \big( H ( {X}^{t,x,u}_{T} - \mathbb{E}_{t} [ {X}^{t,x,u}_{T} ] ) \big) \big] - 1 \\
            & = \kappa \mathbb{E}_{t} [ {X}^{t,x,u}_{T} ]
              - \sum_{j=1}^{\infty } \frac{ (-1)^{j-1} }{ (2j)! } \mathbb{E} [ {H}^{2j} ] \mathbb{E}_{t} \big[ ( {X}^{t,x,u}_{T} - \mathbb{E}_{t} [ {X}^{t,x,u}_{T} ] )^{2j} \big].
\end{aligned}
\label{eq:objective function case IV :eq}
\end{equation}
In fact, when the agent cannot specify an exact penalty function ${S}_{h} (x) = 1 - \cos ( h x )$ because she/he cannot have complete information about the principal or the environment,
it might be an appropriate choice to try the expected penalty function $S(x) = \mathbb{E} [ {S}_{H} (x) ] = 1 - \mathbb{E} [ \cos ( H x ) ]$ with a given distribution of $H$.
Since
\begin{equation*}
  \sum_{j=1}^{\infty } j (2j-1) {\alpha }_{2j-2} ( {y}^{\beta }_{t} ) {\psi }_{ {z}_{2j} } \big( t, \vec{\alpha } ( {y}^{\beta }_{t} ) \big)
= - \frac{1}{2} \sum_{j=1}^{\infty } \frac{ ( - {y}^{\beta }_{t} )^{j-1} }{(2j-2)!!} \mathbb{E} [ {H}^{2j} ]
= - \frac{1}{2} \mathbb{E} \big[ {H}^{2} {e}^{ - \frac{1}{2} {H}^{2} {y}^{\beta }_{t} } \big],
\end{equation*}
the condition \cref{eq:concavity condition for CNEC particular :eq} for $n = \infty $ holds.
Then, for a CNEC/ONEC candidate in the form of \cref{eq:CNEC particular :eq}, due to \cref{eq:linear combination objective function :eq,eq:algebraic equation :eq}, it is supposed to solve
$ {\kappa }^{2} {\theta }_{t}
= \int_{0}^{ {y}_{t} } \big( \mathbb{E} \big[ {H}^{2} {e}^{ - \frac{1}{2} {H}^{2} z } \big] \big)^{2} dz$
for the unique solution ${y}^{\beta }_{t}$. As a result, ${\beta }_{t} = \kappa {B}_{t} | {D}_{t} |^{-2} / \mathbb{E} [ {H}^{2} \exp ( - {H}^{2} {y}^{\beta }_{t} / 2 ) ]$.

Now we are going to show that for \cref{eq:objective function case IV :eq},
\begin{equation}
\tilde{u} ( t,x ) = \bar{u}_{t} = {e}^{ - \int_{t}^{T} {A}_{v} dv } \frac{ \kappa {B}_{t} }{ | {D}_{t} |^{2} } \big( \mathbb{E} \big[ {H}^{2} {e}^{ - \frac{1}{2} {H}^{2} {y}^{\beta }_{t} } \big] \big)^{-1}
                  - \frac{ {F}_{t} }{ {D}_{t} }
\label{eq:NEC in case IV :eq}
\end{equation}
exactly gives a CNEC and an ONEC, and the corresponding EVF is
\begin{equation}
V ( t,x ) = \kappa \Theta ( t,x ) + 1 - \mathbb{E} \big[ {e}^{ - \frac{1}{2} {H}^{2} {y}^{\beta }_{t} } \big].
\label{eq:EVF in case IV :eq}
\end{equation}
On the one hand, let $\Phi ( y, {z}_{1} ) = \kappa y - S ( y - {z}_{1} )$ and
\begin{equation*}
\left\{ \begin{aligned}
{\lambda }^{s,y} ( t,x ) & = \mathbb{E}_{t} \big[ {\Phi }_{ {z}_{1} } \big( {X}^{ t,x, \tilde{u} }_{T}, {m}^{ \tilde{u} }_{1} ( s,y ) \big) \big]
                           = \mathbb{E}_{t} \big[ S' \big( {X}^{ t,x, \tilde{u} }_{T} - {m}^{ \tilde{u} }_{1} ( s,y ) \big) \big], \\
    {\mu }^{s,y} ( t,x ) & = \mathbb{E}_{t} \big[ {\Phi }_{ {z}_{1} {z}_{1} } \big( {X}^{ t,x, \tilde{u} }_{T}, {m}^{ \tilde{u} }_{1} ( s,y ) \big) \big]
                           = - \mathbb{E}_{t} \big[ S'' \big( {X}^{ t,x, \tilde{u} }_{T} - {m}^{ \tilde{u} }_{1} ( s,y ) \big) \big],
\end{aligned} \right.
\end{equation*}
which leads to ${\lambda }^{s,y}_{x} ( t,x ) = - \exp ( \int_{t}^{T} {A}_{v} dv ) {\mu }^{s,y} ( t,x )$ and
\begin{equation*}
  {\mu }^{t,x} ( t,x )
= - \mathbb{E}_{t} \bigg[ {H}^{2} \cos \bigg( H \int_{t}^{T} {D}_{v} {\beta }_{v} d {W}_{v} \bigg) \bigg]
= - \mathbb{E} \big[ {H}^{2} {e}^{ - \frac{1}{2} {H}^{2} {y}^{\beta }_{t} } \big].
\end{equation*}
Then, for \cref{eq:EVF in case IV :eq}, the right-hand side of \cref{eq:Bellman equation in case I :eq} is equal to
\begin{align*}
\max_{ \zeta \in \mathbb{R} }
\bigg\{ - \frac{ {\kappa }^{2} | {B}_{t} |^{2} }{ 2 | {D}_{t} |^{2} \mathbb{E} \big[ {H}^{2} {e}^{ - \frac{1}{2} {H}^{2} {y}^{\beta }_{t} } \big] }
      + \kappa {e}^{ \int_{t}^{T} {A}_{v} dv } \frac{ {B}_{t} }{ {D}_{t} } ( {D}_{t} \zeta + {F}_{t} ) \\
      - \frac{1}{2} \mathbb{E} \big[ {H}^{2} {e}^{ - \frac{1}{2} {H}^{2} {y}^{\beta }_{t} } \big] {e}^{ 2 \int_{t}^{T} {A}_{v} dv } ( {D}_{t} \zeta + {F}_{t} )^{2} & \bigg\} = 0,
\end{align*}
whereby \cref{eq:NEC in case IV :eq} realizes the maximum.
Hence, we obtain the desired conclusion for $( \tilde{u}, V )$.
On the other hand, since the following BSDE arising from \cref{eq:FBSDE particular :eq},
\begin{equation*}
\left\{ \begin{aligned}
d {Y}^{t}_{s} & = - {A}_{s} {Y}^{t}_{s} ds + \mathcal{Y}^{t}_{s} d {W}_{s}, \quad \forall s \in [ t,T ], \\
  {Y}^{t}_{T} & = \kappa - \mathbb{E}_{T} \bigg[ H \sin \bigg( H \int_{t}^{T} {D}_{v} {\beta }_{v} d {W}_{v} \bigg) \bigg],
\end{aligned} \right.
\end{equation*}
admits the solution
\begin{equation*}
\left\{ \begin{aligned}
    {Y}^{t}_{s} {e}^{ - \int_{s}^{T} {A}_{v} dv }
& = \kappa - \mathbb{E}_{s} \bigg[ H \sin \bigg( H \int_{t}^{s} {D}_{v} {\beta }_{v} d {W}_{v} \bigg) {e}^{ - \frac{1}{2} {H}^{2} {y}^{\beta }_{s} } \bigg], \\
    \mathcal{Y}^{t}_{s} {e}^{ - \int_{s}^{T} {A}_{v} dv }
& = - \mathbb{E}_{s} \bigg[ {H}^{2} \cos \bigg( H \int_{t}^{s} {D}_{v} {\beta }_{v} d {W}_{v} \bigg) {e}^{ - \frac{1}{2} {h}^{2} {y}^{\beta }_{s} } \bigg] {D}_{s} {\beta }_{s},
\end{aligned} \right.
\end{equation*}
we obtain $\lim_{s \downarrow t} \mathbb{E}_{t} [ {B}_{s} {Y}^{t}_{s} + {D}_{s} \mathcal{Y}^{t}_{s} ] = {B}_{t} {Y}^{t}_{t} + {D}_{t} \mathcal{Y}^{t}_{t} = 0$.
Moreover, given that $d {Z}^{t}_{s} = - 2 {A}_{s} {Z}^{t}_{s} ds + \mathcal{Z}^{t}_{s} d {W}_{s}$ and
${Z}^{t}_{T} = - \mathbb{E}_{T} [ {H}^{2} \cos ( H \int_{t}^{T} {D}_{v} {\beta }_{v} d {W}_{v} ) ]$ from \cref{eq:FBSDE particular :eq}, we obtain
$\lim_{ s \downarrow t } {Z}^{t}_{s} = {Z}^{t}_{t} = - \mathbb{E}_{t} [ {H}^{2} \exp ( - \frac{1}{2} {H}^{2} {y}^{\beta }_{t} ) ] < 0$.
Hence, we obtain the desired conclusion for $\bar{u}$. So we are done.

\begin{remark}
Let $i$ be the imaginary unit for (generalized) Fourier transform.
For a fairly general even penalty function $S(x)$, its Fourier (cosine) transform
\begin{equation*}
\hat{S} (h) := \int_{ - \infty }^{ + \infty } S(x) {e}^{ - i h x } dx = \int_{ - \infty }^{ + \infty } S(x) \cos (hx) dx
\end{equation*}
is also an even real-valued function. By inverse Fourier (cosine) transform,
\begin{equation*}
S(x) = \frac{1}{ 2 \pi } \int_{ - \infty }^{ + \infty } \hat{S} (h) \cos (hx) dh
     = \sum_{j=0}^{\infty } \frac{ (-1)^{j} {x}^{2j} }{ 2 \pi (2j)! } \int_{ - \infty }^{ + \infty } \hat{S} (h) {h}^{2j} dh.
\end{equation*}
This is a generalization of the form $S(x) = 1 - \mathbb{E} [ \cos (Hx) ]$.
Roughly speaking, one can treat $\hat{S} / 2 \pi $ as a ``probability density function'', or replace the probability measure with $\hat{S} (h) dh / 2 \pi $ in our previous derivation.
For example, for $S(x) = 1 - \exp ( - \frac{1}{2} {x}^{2} )$,
\begin{equation*}
\frac{1}{ 2 \pi } \hat{S} (h) = \delta ( h ) - \frac{1}{ 2 \pi } \int_{ - \infty }^{ + \infty } {e}^{ - i h x - \frac{ {x}^{2} }{2} } dx = \delta ( h ) - \frac{1}{ \sqrt{ 2 \pi } } {e}^{ - \frac{ {h}^{2} }{2} },
\end{equation*}
where $\delta ( \cdot )$ is Dirac delta function.
In terms of determining $\beta $ for the CNEC/ONEC in the form of \cref{eq:CNEC particular :eq}, according to \cref{eq:algebraic equation :eq}, it is supposed to solve
\begin{align*}
  {\kappa }^{2} {\theta }_{t}
= \int_{0}^{ {y}^{\beta }_{t} } \bigg( \int_{ - \infty }^{ + \infty } \frac{ \hat{S} (h) }{ 2 \pi } {h}^{2} {e}^{ - \frac{1}{2} {h}^{2} z } dh \bigg)^{2} dz.
\end{align*}
As a result, if $\int_{ - \infty }^{ + \infty } \hat{S} (h) {h}^{2} \exp ( - {h}^{2} {y}^{\beta }_{t} / 2 ) dh > 0$,
interested readers can arrive at the CNEC/ONEC and its corresponding EVF, as the following:
\begin{align*}
& \tilde{u} ( t,x ) = \bar{u}_{t}
  = {e}^{ - \int_{t}^{T} {A}_{v} dv } \frac{ \kappa {B}_{t} }{ | {D}_{t} |^{2} }
    \bigg( \int_{ - \infty }^{ + \infty } \frac{ \hat{S} (h) }{ 2 \pi } {h}^{2} {e}^{ - \frac{1}{2} {h}^{2} {y}^{\beta }_{t} } dh \bigg)^{-1}
  - \frac{ {F}_{t} }{ {D}_{t} }, \\
& V ( t,x ) = \kappa \Theta ( t,x ) + 1 - \int_{ - \infty }^{ + \infty } \frac{ \hat{S} (h) }{ 2 \pi } {e}^{ - \frac{1}{2} {h}^{2} {y}^{\beta }_{t} } dh.
\end{align*}
Notably, any real function $S$ defined on an interval symmetric about the origin has the decomposition $S(x) = ( S(x) + S(-x) ) / 2 + ( S(x) - S(-x) ) / 2$, where the first part is even and the second part is odd.
According to \Cref{rem:independence for odd central moments}, the state-independent CNEC and the path-independent ONEC for a fairly general penalty function $S(x)$ might be the same as those for $( S(x) + S(-x) ) / 2$,
for which the above Fourier cosine transform method can be used. 
\end{remark}

\section{Concluding remark}
\label{sec:Concluding remark}

We have studied the time-consistent stochastic control problems with higher-order central moments of the terminal value of state process.
On the one hand, seeking a closed-loop Nash equilibrium control is reduced to solving a PDE system,
where the so-called equilibrium value function as well as its Bellman equation is not necessary.
On the other hand, referring to a standard perturbation argument for spike variation, we derive a sufficient maximum principle with a flow of FBSDEs for open-loop Nash equilibrium control.
After the necessary theory is established, we consider the linear control problems, where the objective functional is affine in the mean.
In many cases, the closed-loop and the open-loop Nash equilibrium controls are identical and independent of the state, random path and the preference on odd-order central moments.

\section*{Acknowledgments}
The research of the first author is supported by National Natural Science Foundation of China under Grant 12401611 and CTBU Research Projects Grant 2355010.
The research of the second author is supported by Major Program of the Key Research Institute on Humanities and Social Science of China Ministry of Education under Grant 22JJD790091.
The third author acknowledges the funding support from NSF-DMS Grant 2204795.
The fourth author acknowledges the funding support from the Research Grants Council of the Hong Kong Special Administrative Region Grant PolyU15223419.

\bibliographystyle{apacite}
\bibliography{references}
\end{document}